\documentclass[a4paper]{article}
\usepackage[utf8]{inputenc}
\usepackage{amsmath}
\usepackage{amsfonts}
\usepackage{amsthm}
\usepackage{tikz}
\usepackage{ifthen}
\usepackage{myprobdescr}
\usepackage{todonotes}
\usepackage[nice]{nicefrac}
\usepackage{verbatim}

\usepackage{microtype,ellipsis}

\usepackage[pagebackref,pdfdisplaydoctitle,menucolor=orange!40!black,filecolor=magenta!40!black,urlcolor=blue!40!black,linkcolor=red!40!black,citecolor=green!40!black,colorlinks]{hyperref}
\renewcommand*{\backref}[1]{}
\renewcommand*{\backrefalt}[4]{%
\ifcase #1%
\marginpar{\tiny no cite}
\or
 Cited
on p.~#2.%
\else
  Cited
on pp.~#2.%
\fi
}

\usetikzlibrary{arrows,decorations.pathmorphing,decorations.pathreplacing,backgrounds,positioning,fit}
\usetikzlibrary{shapes}
\usetikzlibrary{arrows,calc}

\usepackage[numbers,sort]{natbib}
\makeatletter
\def\NAT@spacechar{~}
\makeatother

\newtheorem{definition}{Definition}
\newtheorem{theorem}{Theorem}
\newtheorem{lemma}{Lemma}

\newtheorem{observation}{Observation}
\newtheorem{proposition}{Proposition}
{\rm\bfseries}{\itshape}

\newcommand{\dash}{\nobreakdash-\hspace{0pt}}
\newcommand{\DAGR}{\textsc{DAG Realization}\xspace}
\newcommand{\TPart}{\textsc{3\dash Partition}\xspace}
\newcommand{\SF}{\textsc{Sequence Filling}\xspace}

\renewcommand{\S}{\ensuremath{\mathcal{S}}}
\newcommand{\T}{\ensuremath{\phi}}
\renewcommand{\P}{\mathcal{P}}
\newcommand{\A}{\mathcal{A}}
\newcommand{\pot}{\ensuremath{\Delta^2}}
\newcommand{\superTypes}{\ensuremath{\mathcal{T}^S}}
\newcommand{\N}{\ensuremath{\mathbb{N}}}

\DeclareMathOperator{\pos}{pos}
\newcommand{\nRepOr}{non-repeating ordering\xspace}

\newcommand{\tOrd}{topological ordering\xspace}
\newcommand{\TopOr}{realizing \tOrd}
\newcommand{\pTopOr}{partial \TopOr}

\title{\textbf{NP-Hardness and Fixed-Parameter Tractability of Realizing Degree Sequences with Directed Acyclic Graphs}}
\author{
Sepp Hartung and André Nichterlein \medskip \\  
  Institut f\"ur Softwaretechnik und Theoretische Informatik, \\ TU Berlin, Berlin, Germany \\
\small{\texttt{\{sepp.hartung,andre.nichterlein\}@tu-berlin.de}}}
\date{}

\usepackage{etoolbox}

\newcounter{Bew1}
\newcounter{Bew2}

\begin{document}
\def\sectionautorefname{Section}
\def\subsectionautorefname{Subsection}

\maketitle

\begin{abstract}

In graph realization problems one is given a degree sequence and the task is to decide whether there is a graph whose vertex degrees match to the given sequence.
This realization problem is known to be polynomial-time solvable when the graph is directed or undirected.
In contrary, we show NP-completeness for the problem of realizing a given sequence of pairs of positive integers (representing indegrees and outdegrees) with a \emph{directed acyclic graph}, answering an open question of Berger and Müller-Hannemann [FCT 2011].
Furthermore, we classify the problem as fixed-parameter tractable with respect to the parameter ``maximum degree''.

\end{abstract}

\section{Introduction}

\citet{BM11} introduced the following problem:

\begin{center}
	\begin{minipage}{0.95\textwidth}
	\defDecprob{DAG Realization}
		{A multiset~$\S = \left\{ \binom{a_1}{b_1}, \ldots, \binom{a_n}{b_n} \right\}$ of integer pairs with~$a_i, b_i \geq 0$.}
		{Is there a directed acyclic graph (without parallel arcs and self-loops) that admits a labeling of its vertex set~$\{v_1, \ldots, v_n\}$ such that for all $v_i \in V$ the indegree is~$a_i$ and the outdegree is~$b_i$?}
	\end{minipage}
\end{center}
If the \emph{degree sequence}~$\S$ is a yes-instance, then~$\S$ is called \emph{realizable} and the corresponding directed acyclic graph (dag for short) $D$ is called a \emph{realizing dag} for~$\S$.
\citet{BM11} showed that this problem is polynomial-time solvable for special types of degree sequences, but left the complexity of the general problem as their main open question.
We answer this question by showing that \DAGR{} is NP-complete. 
Moreover, on the positive side we classify \DAGR as fixed-parameter tractable with respect to the parameter maximum degree $\Delta:=\max\{a_1,b_1,\ldots,a_n,b_n\}$. The corresponding algorithm actually constructs for yes-instances a realizing dag.

\paragraph{\bf Related Work.} It is known for a long time that deciding whether a given degree sequence (a multiset of positive integers) is realizable with an \emph{undirected graph} is polynomial-time solvable.
There are characterizations for realizable degree sequences due to \citet{EG60} and algorithms by \citet{Hav55} and \citet{Hak62}.
In the case, where one asks whether there is a \emph{directed graph} realizing the given degree sequence (a multiset of positive integer pairs), has also been intensively studied: See \citet{Gal57,Rys57,Ful60,Che66} for characterizations of digraph realizations and \citet{KW73} for polynomial-time algorithms.
The problem of realizing degree sequences has also been studied in context of (loop-less) multigraphs, where the aim is to minimize or maximize the number of multi-edges~\cite{HWW08}.

\section{Preliminaries}

We set $\N:=\{0,1,2,\ldots\}$. We denote with $\uplus$ the multiset sum (e.g~$\{1,1\} \uplus \{1,2\} = \{1,1,1,2\}$). 

A \emph{parameterized problem} $(I,k)$, consisting of the problem instance~$I$ and the parameter~$k\in \N$, is \emph{fixed-parameter tractable} if it can be solved in $f(k)\cdot n^c$ time. Thereby, $f$~is a computable function solely depending on~$k$ and~$c\in \N$ is a constant independent from~$I$ and~$k$. For a more detailed introduction to parameterized algorithmics and complexity we refer to the monographs \cite{DF99, FG06,Nie06}.

We denote directed graphs by~$D=(V,A)$ with vertex set~$V$ and arc set~$A\subseteq V\times V$. The indegree of~$v\in V$ is denoted by~$d^-(v)$ and the outdegree by~$d^+(v)$. Correspondingly, for a degree sequence~$\S$ and an \emph{element} $s \in \S$ with~$s = \binom{a}{b}$, we set~$d^-(s) := a$ and~$d^+(s) := b$.

A directed graph~$D=(V,A)$ is a \emph{dag} if it does not contain a cycle.
A cycle is a vertex sequence $v_1,\ldots,v_l$ such that for all $1\le i<l:$ $(v_i,v_{i+1})\in A$ and $(v_l,v_1)\in A$. 
Each dag~$D$ admits a \emph{\tOrd}, that is, an ordering of all its vertices $v_1,\ldots,v_n$ such that for all arcs $(v_i,v_j)\in A$ it holds that $i<j$.
Consequently, for a realizing dag we call a corresponding \tOrd a \emph{\TopOr}. 

We use the \emph{opposed order}~$\leq_{\rm{opp}}$ for the elements of a degree sequence~$\S$, as introduced by \citet{BM11}:
\begin{definition}
	$\binom{a_1}{b_1} \leq_{\rm{opp}} \binom{a_2}{b_2} \iff (a_1 \leq a_2 \wedge b_1 \geq b_2)$ 
%
\end{definition}
Note that there might be elements in the degree sequence~$\S$ that are not comparable with respect to the opposed order.
However, we can always assume that a realization does not collide with the opposed order and thus \DAGR is polynomial-time solvable in case of all elements of~$\S$ are comparable.

\begin{lemma}[{\cite[Corollary 3]{BM11}}] \label{obs:oppOrderTopoplogical}
	Let~$\S=\left\{\binom{a_1}{b_1},\ldots,\binom{a_n}{b_n}\right\}$ be a realizable degree sequence. Then, there exists a \TopOr~$\phi$ such that for all $1\le i,j\le n$ with $s_i=\binom{a_i}{b_i}\leq_{\rm{opp}}\binom{a_j}{b_j} = s_j$ and $s_i\neq s_j$, it holds that in~$\phi$ the position of the vertex that corresponds to~$s_i$ is smaller than the position of the vertex that corresponds to~$s_j$.
\end{lemma}
Our paper is organized as follows: The next section contains the proof of the NP-hardness and in \autoref{sec:fpt} we show that \DAGR is fixed-parameter tractable with respect to the parameter maximum degree~$\Delta$.

\section{NP-Completeness}\label{sec:NP-Hardness}

In this section we show the NP-hardness of \DAGR{} by giving a polynomial-time many-to-one reduction from the strongly NP-hard problem \TPart{}~\cite{GJ79}: 
\begin{center}
\begin{minipage}{0.95\textwidth}
\defDecprob{\TPart}
	{A sequence~$\A = a_1, \ldots, a_{3m}$ of~$3m$ positive integers and an integer~$B$ with~$\sum_{i=1}^{3m} a_i = mB$ and~$\forall i: B/4 < a_i < B/2$.}
	{Is there a partition of the $3m$ integers from~$\A$ into~$m$ disjoint triples such that in every triple the three elements add up to~$B$?}
\end{minipage}
\end{center}
%
%
This section is organized as follows: First we describe the construction of our reduction and explain the idea of how it works. 
Then, we prove the correctness in the remainder of the section.


\paragraph{Construction.}
Given an instance~$(\A,B)$ of \TPart{}, we construct an equivalent instance~$\S$ of \DAGR{} as follows:
$$ \S := X_0, X_1, \ldots, X_m, \alpha_1, \alpha_2, \ldots, \alpha_{3m} $$
where~$\alpha_i = \binom{a_i}{a_i}$, $1 \le i \le 3m$. The~$X_i$, $0 \le i \le m$, are subsequences which we formally define after giving the idea of the construction.
We call an element from a subsequence~$X_i$ an \emph{$x$-element} and the~$\alpha_j$ are called \emph{$a$-elements}.
In a realizing dag~$D$ the vertices realizing $x$-elements are called \emph{$x$-vertices} and the vertices realizing $a$-elements are called \emph{$a$-vertices}.

The intuition of the construction is that a dag~$D$ realizing~$\S$ (if it exists) looks as follows:
The vertices realizing elements of a subsequence $X_i$, $0 \le i \le m$, form a ``block'' in a \TopOr~$\T$. 
These blocks are a skeletal structure in any \TopOr.
There are~$m$ ``gaps'' between these blocks of $x$-vertices.
The construction ensures that these gaps are filled with $a$-vertices and,  moreover, the indegree and outdegree of all the $a$-vertices in a gap sum up to~$B$. Hence, these $m$ gaps require to partition the $a$-vertices into~$m$ sets where each them has in total in- and outdegree~$B$ and, thus, correspond to a solution for the \TPart instance where we reduce from.
In the reverse direction, for each triple in a solution of a \TPart instance the corresponding $a$-vertices will be used to fill up one gap.
See \autoref{fig:exampleRealization} for an example of the construction.
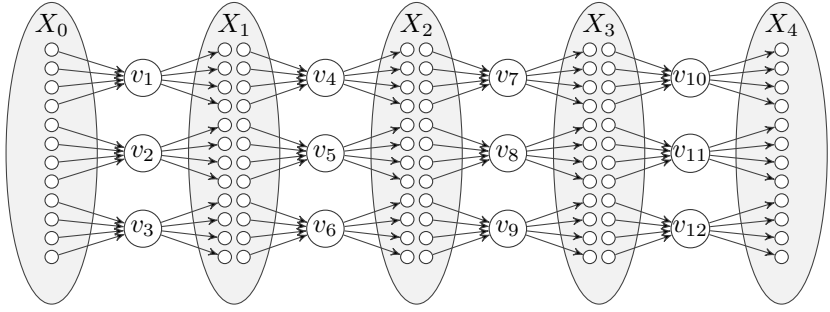
\begin{figure}[t]
	\begin{center}
		\def\layersep{2.4cm}
		\def\B{12}
		\begin{tikzpicture}[draw=black!75, scale=1,->,>=stealth']
			\tikzstyle{vertex}=[circle,draw=black!80,minimum size=14pt,inner sep=0pt]
			\tikzstyle{x-vertex}=[circle,draw=black!80,minimum size=5pt,inner sep=0pt,fill=white]

			\draw[-,fill=black!5] (0,-2) ellipse (0.6 cm and 2 cm);

			\foreach \y in {1,...,\B}{
				\pgfmathparse{int(\y+\B)}  \let\nr\pgfmathresult
				\node[x-vertex] (x-\nr) at (0,-\y * 0.25 - 0.375) {};
			}
			\node at (0 * \layersep +0.1,-0.3) {$X_0$};

			\foreach \x in {1,...,4}{
				\pgfmathparse{int(\x+1)}  \let\xx\pgfmathresult
				\draw[-,fill=black!5] (\layersep * \x,-2) ellipse (0.6 cm and 2 cm);
				\node at (\x * \layersep + 0.1,-0.3) {$X_\x$};
				\ifthenelse{\x<4}{
					\foreach \y in {1,...,\B}{
						\pgfmathparse{int(\y+\B*2*(\x))}  \let\nr\pgfmathresult
						\ifthenelse{\x=1}{
							\ifthenelse{\y=1}{
								\node[x-vertex]	(x-\nr) at (\x * \layersep - 3.5,-\y * 0.25 - 0.375) {};
							} {
								\node[x-vertex]	(x-\nr) at (\x * \layersep - 3.5,-\y * 0.25 - 0.375) {};
							}
						} {
							\node[x-vertex]	(x-\nr) at (\x * \layersep - 3.5,-\y * 0.25 - 0.375) {};
						}
					}
					\foreach \y in {1,...,\B}{
						\pgfmathparse{int(\y+\B*2*(\x)+\B)}  \let\nr\pgfmathresult
						\node[x-vertex]	(x-\nr) at (\x * \layersep + 3.5,-\y * 0.25 - 0.375) {};
					}
				}{
					\foreach \y in {1,...,\B}{
						\pgfmathparse{int(\y+\B*2*(\x))}  \let\nr\pgfmathresult
						\node[x-vertex]	(x-\nr) at (\x * \layersep,-\y * 0.25 - 0.375) {};
					}

				}
				\foreach \y in {1,...,3}{
					\pgfmathparse{int((\y-1)*4+1 +\B*2*(\x))}			\let\ya\pgfmathresult
					\pgfmathparse{int((\y)*4 +\B*2*(\x))}					\let\yb\pgfmathresult

					\pgfmathparse{int((\y-1)*4+1+\B*2*(\x-1)+\B)}		\let\za\pgfmathresult
					\pgfmathparse{int((\y)*4 +\B*2*(\x-1)+\B)}			\let\zb\pgfmathresult
			
					\pgfmathparse{int(\y+3*(\x-1))}	\let\nr\pgfmathresult
					\node[vertex] 	(a-\nr) at (\x * \layersep - 0.5*\layersep,-\y) {$v_{\nr}$};
					\foreach \i in {\ya,...,\yb}{
						\path (a-\nr) edge (x-\i);
					}
					\foreach \i in {\za,...,\zb}{
						\path (x-\i) edge (a-\nr);
					}
				}
			}
		\end{tikzpicture}
	\end{center}
	\caption{A schematic representation of a dag that realizes a degree sequence~$\S$ that is constructed from a \textsc{3-Partition} instance with $B=12$ and~$m=4$. There are five blocks marked by the gray ellipses and four gaps between them.
	In each gap there are three $a$-vertices, altogether having in- and outdegree~$B$. 
	The sets~$X_i$, $1\le i \le 3$, are partitioned into two parts of size~$B$. The vertices in the left part have $B$~ingoing arcs from the $a$-vertices that fill the gap between $X_{i-1}$ and~$X_i$. Correspondingly, the vertices in the right part have~$B$ outgoing arcs to the $a$-vertices that fill the gap between~$X_i$ and $X_{i+1}$.
	Consequently, the first and the last block~$X_0$ and~$X_4$ are of size~$B$.
	The in- and outdegree of the $a$-vertices in each triple sum up to~$B$.
	}
	\label{fig:exampleRealization}
\end{figure}


To achieve the mentioned skeletal structure of the subsequences~$X_0, \ldots, X_m$, we require the corresponding $x$-vertices to form a \emph{complete dag}: A dag with $n$~vertices and~$\binom{n}{2}$ arcs that realizes the degree sequence $\left\{ \binom{0}{n-1},\binom{1}{n-2},\ldots, \binom{n-1}{0}\right\}$. Observe that there is only one dag realizing such a sequence and, furthermore, such a complete dag admits only one \tOrd.

Now, we complete the reduction by defining the subsequences~$X_0, \ldots, X_m$.
As indicated in \autoref{fig:exampleRealization}, $X_0$~and~$X_m$ contain~$B$ elements and the other subsequences contain~$2B$ elements.
The subsequence~$X_0$ consists of the elements~$x_0^0$, $x_0^1$, $\ldots$, $x_0^{B-1}$. 
This subsequence corresponds to the $x$-vertices~$v_0^0,$ $\ldots,$ $v_0^{B-1}$ forming the first block in a realizing dag for $\S$. 
Remember that the $x$-vertices are supposed to form a complete dag. 
To achieve this, $v_0^j$ has $(B - 1 - j)$ outgoing arcs to~$v_0^{j+1}, \ldots, v_0^{B-1}$ and $(m-1)2B+B=(2m-1)B$ outgoing arcs to the $x$-vertices in the subsequent blocks.
Furthermore,~$v_0^j$ has~$j$ ingoing arcs from the $x$-vertices~$v_0^{0}, \ldots, v_0^{j-1}$.
Finally, each $x$-vertex in~$v_0^0, \ldots, v_0^{B-1}$ has one outgoing arc to one of the three subsequent $a$-vertices.
Hence, the corresponding $x$-element of~$v_0^j$ is as follows:
\begin{align*}
	x_0^j & :=
		\binom{j}{(B - 1 - j) + (2m-1)B + 1} = 
		\binom{j}{2mB - j}.
\end{align*}
Analogously, the subsequence~$X_m$ consists of~$B$ elements~$x_m^0, x_m^1, \ldots, x_m^{B-1}$ defined as follows:
\begin{align*}
	x_m^j & :=
		\binom{(2m-1)B + j + 1}{B - 1 - j}.
\end{align*}
For~$0<i<m$, the subsequence~$X_i$ consists of~$2B$ elements~$x_i^0, x_i^1, \ldots, x_i^{2B-1}$. 
Let~$v_i^0, \ldots, v_i^{2B-1}$ denote the corresponding $x$-vertices. 
Then, $v_i^j$ has~$(i-1)2B+B = (2i-1)B$ ingoing arcs from the~$x$-vertices in the preceding blocks and $j$~ingoing arcs from~$v_i^0, \ldots, v_i^{j-1}$.
Furthermore, $v_i^j$ has~$(m-i-1)2B + B = (2m-2i-1)B$ outgoing arcs to the subsequent blocks and~$2B-1-j$ outgoing arcs to~$v_i^{j+1}, \ldots, v_i^{B-1}$.
Finally, if~$j < B$, then $v_i^j$ has an ingoing arc from one of the three preceding $a$-vertices.
Otherwise, if~$j \ge B$, then $v_i^j$ has an outgoing arc to one of the three subsequent $a$-vertices.
Hence, the corresponding $x$-element of~$v_i^j$ is as follows:
\begin{align*}
	x_i^j & := 
			\binom{(2i-1)B + j + 1}{(2m-2i+1)B -1-j}		  & \text{ if } j < B, \\
	x_i^j & := 
			\binom{(2i-1)B + j}{(2m-2i+1)B -j}   & \text{ if } j \ge B.
\end{align*}
Observe that the strong NP-hardness of \TPart is essential to prove the polynomial running time of the reduction:
The size of the constructed \DAGR instance is upper-bounded by a polynomial in the values of the integers in~$\A$.
Since \TPart is strongly NP-hard, 
it remains NP-hard when the values of the integers in~$\A$ are bounded by a polynomial in the input size.
Hence, the size of the \DAGR instance is polynomially bounded in the size of the \TPart instance. Clearly, the construction can be computed in polynomial time.

\paragraph{Correctness.} In the following, we prove the correctness of the construction given above. Therefore, throughout this subsection let $(\A,B)$ be an instance of \TPart and let $\S$ be the corresponding degree sequence formed by the construction above.


\begin{lemma} \label{lem:hinrichtung}
	If $(\A,B)$ is a yes-instance of \TPart{}, then~$\S$ is a yes-instance of \DAGR{}.
\end{lemma}

\begin{proof} 
	We prove that if \TPart is a yes-instance, then there exists a realizing dag for~$\S$ as described above and pictured in \autoref{fig:exampleRealization}.

 	Let~$\pi$ be a permutation of the sequence~$\A$ such that~$a_{\pi(3i+1)} + a_{\pi(3i+2)} + a_{\pi(3i+3)} = B$ for all~$0 \leq i < m$.
 	Since~$(\A,B)$ is a yes-instance of \TPart{} such a permutation exists.
	We now construct a realizing dag~$D = (V,A)$.
	The degree sequence~$\S$ and, hence, a realizing dag~$D$ consists of~$|V| = B+(m-1)2B+B+3m = 2mB+3m$ vertices.
	We group~$V$ into~$2m+1$ disjoint vertex sets~$V = V^b_0 \cup V^b_1 \cup \ldots \cup V^b_m \cup V^t_1 \cup \ldots \cup V^t_m$ with~$V^b_i = \{v_i^0, v_i^1, \ldots, v_i^{2B-1}\}$ for all~$1 \leq i < m$ and $V^t_j = \{u_{\pi(3j+1)}$, $u_{\pi(3j+2)}, u_{\pi(3j+3)}\}$ for all~$1 \leq j \le m$.
	The first set is $V^b_0 = \{v_0^0, v_0^1, \ldots, v_0^{B-1}\}$ and~$V^b_m$ contains the last~$B$ vertices~$v_m^0, \ldots, v_m^{B-1}$.

	Each vertex~$v_i^j$ realizes the $x$-element~$x_i^j$.
	Each vertex~$u_i$ realizes the $a$-element~$\alpha_i$.
	The vertex sets~$V^b_i$ form the blocks denoted by the ellipses in \autoref{fig:exampleRealization}.
	The vertex sets~$V^t_j$ correspond to the triples of $a$-vertices filling the gaps between the blocks.
	By construction the indegrees and also the outdegrees of the vertices in each~$V^t_j$ add up to~$B$.

	We now describe how the vertices are connected with arcs:
	The $x$-vertex~$v_i^j$ has an outgoing arc to every vertex of~$V_\ell^b$, $\ell > i$,
	and an outgoing arc to all the ``following'' vertices in his block, that is, the $x$-vertices~$v_i^\ell$ with~$\ell > j$.
	If~$0 \leq j \leq B-1$ and~$0 < i \leq m$, then~$v_i^j$ has one ingoing arc from one of the $a$-vertices of~$V^b_{i-1}$.
	If~$B \leq j \leq 2B-1$ and~$0 < i < m$ or~$0 \leq j \leq B-1$ and~$i=0$, then~$v_i^j$ has one outgoing arc to one of the $a$-vertices of~$V^b_{i}$.
	Since the sum of the indegrees and the sum of the outdegrees in each vertex set~$V^t_j$ adds up to~$B$, the arcs between $a$-vertices and $x$-vertices can be set such that each $a$-vertex~$u_i$ has~$a_i$ ingoing and outgoing arcs.
	This completes the description of~$D$. 
	Clearly,~$D$ is a dag. 
	Hence, it remains to show that~$D$ realizes~$\S$.

	The indegree of~$v_i^j$, $1 \le i <  m$, is as follows:~$v_i^j$ has ingoing arcs from the~$2B(i-1) + B = (2i-1)B$ vertices realizing the elements in~$X_0, X_1, \ldots, X_{i-1}$, from the $j$~vertices~$v_i^0, \ldots, v_i^{j-1}$, and from one $a$-vertex in~$V_{i-1}$ if~$0 \leq j < B$.
	Altogether, this gives an indegree of~$(2i-1)B + j + 1$ if~$0 < j < B$ or~$(2i-1)B + j$ if~$B \leq j < 2B$.

	The outdegree of~$v_i^j$, $1 \le i <  m$, is as follows:~$v_i^j$ has outgoing arcs to the~$2B(m-i-1) + B = (2m-2i-1)B$ vertices realizing the elements in~$X_{i+1}, X_{i+2}, \ldots, X_{m}$, to the~$2B-1-j$ vertices~$v_i^{j+1}, \ldots, v_i^{2B-1}$, and to one $a$-vertex in~$V^t_{i}$ if~$2B > j \geq B$.
	Altogether this gives an outdegree of~$(2m-2i+1)B-j-1$ if~$0 \le j < B$ or~$(2m-2i+1)B-j$ if~$B \le j < 2B$.
	Hence the $x$-vertex~$v_i^j$ fulfills the degree constraints of the $x$-element~$x_i^j$ (special cases of~$i \in \{0,m\}$ follow analogously).

	Each $x$-vertex of~$\{v_{i+1}^0, \ldots, v_{i+1}^{B-1}\}$ has one ingoing arc from one of the $a$-vertices~$u_{\pi(3i+1)},u_{\pi(3i+2)},u_{\pi(3i+3)}$ of~$V^t_i$.
	Hence, the total number of outgoing arcs of~$u_{\pi(3i+1)},u_{\pi(3i+2)}$, and~$u_{\pi(3i+3)}$ is~$B$.
	Each $x$-vertex of~$\{v_{i}^{B}, \ldots, v_{i}^{2B-1}\}$ has one outgoing arc to one of the $a$-vertices~$u_{\pi(3i+1)},u_{\pi(3i+2)},u_{\pi(3i+3)}$ of~$V^t_i$.
	Hence, the total number of ingoing arcs of~$u_{\pi(3i+1)},u_{\pi(3i+2)}$ and~$u_{\pi(3i+3)}$ is~$B$.
	Since~$a_{\pi(3i+1)}+a_{\pi(3i+2)}+a_{\pi(3i+3)} = B$, the $a$-vertices~$u_{\pi(3i+1)},u_{\pi(3i+2)}$, and $u_{\pi(3i+3)}$ fulfill the degree constraints of~$\alpha_{\pi(3i+1)}, \alpha_{\pi(3i+2)}$, and~$\alpha_{\pi(3i+3)}$.

	Overall, each $a$-vertex~$u_i$ has indegree and outdegree equal to~$a_i$ and each vertex~$v_i^j$ fulfills the degree constraints of~$x_i^j$.
\end{proof}
%
\looseness=-1 To show the reverse direction, we first need some observations. 

\begin{observation} \label{obs:aVerticesIS}
	In any dag~$D$ realizing~$\S$, the $a$-vertices form an independent set and the $x$-vertices form a complete dag.
\end{observation}

\begin{proof}
	The number~$d^-(X)$ of ingoing arcs to all $x$-vertices is:
	\begin{align*}
		d^-(X) {} = & \sum_{j=0}^{B-1} d^-(x_0^j) + \sum_{i=1}^{m-1}\sum_{j=0}^{B-1} d^-(x_i^j) + \sum_{i=1}^{m-1}\sum_{j=B}^{2B-1} d^-(x_i^j) + \sum_{j=0}^{B-1} d^-(x_m^i)\\
		= {} & \sum_{j=0}^{B-1} j  + \sum_{i=1}^{m-1} \sum_{j=0}^{B-1} ( (2i-1)B + j + 1 )  \\  & + \sum_{i=1}^{m-1} \sum_{j=B}^{2B-1} ( (2i-1)B + j ) + \sum_{j=0}^{B-1} ((2m-1)B + j + 1)\\
		= {} & 2 m^2 B^2. \\
	\end{align*}
	Note that~$d^-(X)$ is equal to the number~$d^+(X)$ of outgoing arcs from all $x$-vertices.
	The number of $a$-vertices is~$3m$ and the number of $x$-vertices is~$2mB$.
	Hence, the number~$\xi$ of arcs connecting two $x$-vertices is at most:
	\begin{equation*}
		\xi  = \frac{1}{2} 2mB(2mB-1) = 2m^2B^2 - mB.
	\end{equation*}
	As a consequence, there are at least~$d^-(X) - \xi = mB$ arcs going from an $a$-vertex to an $x$-vertex.
	Since~$mB = \sum_{i=1}^{3m} a_i$ is the number of outgoing arcs from the $a$-vertices, all outgoing arcs from $a$-vertices go to $x$-vertices.
	Thus, in any realizing dag $D$ the $a$-vertices form an independent set and the number of arcs that connect two $x$-vertices is exactly~$\xi$.
	Hence, the $x$-vertices form a clique in the underlying undirected graph.
\end{proof}
The next observation shows that for a realizable degree sequence~$\S$ there exists a realization~$D$ with a \tOrd of the vertices such that the $x$-vertices are ordered as follows:
$$x_0^0, x_0^1, \ldots, x_0^{B-1}, x_1^0, \ldots, x_1^{2B-1}, x_2^0, \ldots, x_2^{2B-1}, x_3^0, \ldots, x_{m-1}^{2B-1}, x_m^0, \ldots, x_m^{B-1}.$$


\begin{observation} \label{obs:xVerticesTopOr}
	If~$\S$ is realizable, then there exists a \TopOr~$\T$ such that in~$\T$ for all~$i < j$ the vertex realizing~$x_\ell^i$ is ahead of the vertex realizing~$x_\ell^j$ and for all $0 \le h < k \le m$ the vertices realizing elements of~$X_h$ are ahead of the vertices realizing elements of~$X_k$.
\end{observation}

\noindent\textit{Proof.}
	We first show that for all~$i < j$ the vertex realizing~$x_\ell^i$ is ahead of the vertex realizing~$x_\ell^j$:
	Let~$i$ and~$j$ be two integers with~$i < j$.
	By \autoref{obs:oppOrderTopoplogical} it suffices to show that~$x_\ell^i \leq_{\rm{opp}} x_\ell^j$, for $0 \leq \ell \leq m$.
	That is, it suffices to show~$d^-(x_\ell^i) - d^-(x_\ell^j) \leq 0$ and~$d^+(x_\ell^i) - d^+(x_\ell^j) \geq 0$, where~$d^-(x_\ell^i)$ ($d^+(x_\ell^i)$) is the indegree (outdegree) of the $x$-vertex realizing~$x_\ell^i$.
	This is shown in the following case distinction:
	\begin{description}
		\item[Case~$\ell = 0$ ($0 \leq i < j \leq B-1$):]
			\begin{align*}
				d^-(x_0^i) - d^-(x_0^j) 	& = i - j < 0 \\
				d^+(x_0^i) - d^+(x_0^j)
							& = 2mB - i - 2mB+j \\
							& = j-i > 0
			\end{align*}
		\item[Case~$0 < \ell < m$ ($0 \leq i < j \leq 2B-1$):]
			\begin{align*}
				d^-(x_\ell^i) - d^-(x_\ell^j)
						\le {} & (2\ell-1)B+i+1 - ((2\ell-1)B+j) \\
						= {} & i+1-j 
						 \leq 0 \\
				d^+(x_\ell^i) - d^+(x_\ell^j)
						\ge {} & (2m-2\ell+1) - i - 1 - ((2m-2\ell+1) - j) \\
						= {} & j - 1 - i  \geq 0
			\end{align*}
		\item[Case~$\ell = m$ ($0 \leq i < j \leq B-1$):]
			\begin{align*}
				d^-(x_m^i) - d^-(x_m^j) 	& = (2m-1)B + i + 1 - ((2m-1)B+j+1) \\
													& = i-j < 0 \\
				d^+(x_m^i) - d^+(x_m^j)		& = B - 1 - i - (B - 1 -j) \\
													& = j - i > 0   
			\end{align*}
	\end{description}
%
%
	We now show the second part: for all $0 \le h < k \le m$ it holds that in~$\T$ the vertices realizing elements of~$X_h$ are ahead of the vertices realizing elements of~$X_k$.
	By \autoref{obs:oppOrderTopoplogical} and transitivity of~$\leq_{\rm{opp}}$ it remains to show that (1)~$x_0^{B-1} \leq_{\rm{opp}} x_{1}^0$ and (2)~$x_\ell^{2B-1} \leq_{\rm{opp}} x_{\ell+1}^0$ for all~$0<\ell<m$:
	\begin{description}
		\item[(1):]
			\begin{align*}
				d^-(x_{0}^{B-1}) - d^-(x_{1}^{0})
					= {}& B-1 - ((2-1)B+1) \\
					= {}& -2 < 0 \\
				d^+(x_{0}^{B-1}) - d^+(x_{1}^{0})
					= {}& 2mB - (B-1) - ((2m-2+1)B-1)  \\
					= {}& 2 > 0
			\end{align*}
		\item[(2):]
			\begin{align*}
				d^-(x_{\ell}^{2B-1}) - d^-(x_{\ell+1}^{0})
					= {}& (2\ell - 1)B + 2B-1 - ((2(\ell+1)-1)B + 1)  \\
					= {}& -2 < 0 \\
				d^+(x_{\ell}^{2B-1}) - d^+(x_{\ell+1}^{0})
					= {} & (2m-2\ell+1)B - 2B+1 \\ & - ((2m-2(\ell+1)+1)B-1)\\
					= {} & 2 > 0   \tag*{\qed}
			\end{align*}
	\end{description}
%
%
%
With \autoref{obs:aVerticesIS} and~\ref{obs:xVerticesTopOr} we can prove the next lemma, which completes the proof of the correctness of our reduction.

\begin{lemma}\label{lem:rueckrichtung}
	If~$\S$ is a yes-instance of \DAGR{}, then $(\A,B)$ is a yes-instance of \TPart{}.
\end{lemma}

\begin{proof}
	Let~$D = (V,A)$ be the realization of~$\S$ with a \tOrd~$\T$.
	Let~$v_i^j$ be the $x$-vertex realizing~$x_i^j$ and let~$u_i$ be the $a$-vertex realizing~$a_i$.
	Furthermore, $\pos_\T(v)$ denotes the position of~$v$ in the \tOrd~$\T$.
	Since~$\S$ is a yes-instance, we can assume by \autoref{obs:xVerticesTopOr} that~$\pos_\T(v_i^j) < \pos_\T(v_i^\ell)$ for~$j<\ell$ and that~$\pos_\T(v_i^j) < \pos_\T(v_{k}^\ell)$ for~$i<k$.

	From \autoref{obs:aVerticesIS} it follows that none of the $x$-vertices~$v_0^0, v_0^1, \ldots, v_0^{B-1}$ has an ingoing arc from an $a$-vertex, but each has one outgoing arc to an $a$-vertex.
	Hence, we can assume that~$\pos_\T(u_i) > \Phi(v_0^{B-1})$ for all~$1\leq i \leq 3m$.
	Observe that each $x$-vertex~$v_1^0, v_1^1, \ldots, v_1^{B-1}$ has one ingoing arc from an $a$-vertex and no outgoing arc to an $a$-vertex.
	Hence, we can assume that there are $a$-vertices~$u_{i_1}, u_{i_2}, \ldots, u_{i_\ell}$ with~$\pos_\T(v_0^{B-1}) < \pos_\T(u_{i_j}) < \pos_\T(v_1^0)$ and~$\sum_{j=1}^\ell a_{i_j} = B$.
	Since~$B/4 < a_j < B/2$ for all~$1\leq j \leq 3m$, it follows that~$\ell = 3$.

	The vertices~$v_1^B, \ldots, v_1^{2B-1}$ also have no ingoing arc from an $a$-vertex but each of these vertices has an outgoing arc to an $a$-vertex.
	Also, each of the vertices~$v_2^0, \ldots, v_2^{B-1}$ needs one ingoing arc from an $a$-vertex.
	So, again, we can assume that in the \tOrd~$\T$ of~$D$ there are three $a$-vertices between~$v_1^{2B-1}$ and~$v_2^0$ such that their indegrees and also their outdegrees sum up to~$B$.
	Analogously, it follows for all~$1\leq i < m$ that there are three $a$-vertices~$u_{j_1^i}, u_{j_2^i}, u_{j_3^i}$ with~$\pos_\T(v_i^{2B-1}) < \pos_\T(u_{j_1^i})< \pos_\T(u_{j_2^i})< \pos_\T(u_{j_3^i}) < \pos_\T(v_{i+1}^0)$ and~$\sum_{\ell = 1}^3 a_{j_\ell^i} = B$.
	Hence, $(\A,B)$ is a yes-instance of \TPart{}.
\end{proof}
Our construction together with \autoref{lem:hinrichtung} and \autoref{lem:rueckrichtung} yields the NP-hardness of \DAGR{}.
Containment in NP is easy to see: Guessing an~$n$-vertex dag and checking whether or not it is a realization for~$\S$ is clearly doable in polynomial time.
Hence, we arrive at the following theorem.

\begin{theorem}
	\DAGR{} is NP-complete.
\end{theorem}
%
%
\citet{BM11} gave an polynomial-time algorithm for \DAGR{} if the degree sequence can be ordered with respect to the opposed order.
Hence, one may search for other polynomial-time solvable special cases.
One way to identify such special cases is to have a closer look on NP-hardness proofs and to check whether certain ``quantities'' need to be unbounded in order to make the proof (many-to-one reduction) work \cite{KNU10,Nie10}.
In our NP-hardness proof the maximum degree~$\Delta$ is unbounded.
We show in the next section that \DAGR is polynomial-time solvable for constant maximum degree. Indeed, we can even show fixed-parameter tractability with respect to the parameter~$\Delta$.

\section{Fixed-Parameter Tractability} \label{sec:fpt}
\newcommand{\w}{\omega}
Denoting the maximum degree in a degree sequence by $\Delta$, in this section we show that \DAGR is fixed-parameter tractable with respect to the parameter~$\Delta$.
To describe the basic idea that our fixed-parameter algorithm is based on, we need the following definition.
\begin{definition}\label{def:potential}
	Let $\T=v_1,v_2,\ldots,v_n$ be a \tOrd for a dag~$D$. For all $1\le i\le n$, the \emph{potential} at position~$i$  is a vector $p^\T_i\in \N^\Delta$ where $p^\T_i[l]$ for $1\le l\le \Delta$ is the number of vertices in the subsequence $v_1,\ldots,v_i$ that have in~$D$ at least~$l$ neighbors in the subsequence $v_{i+1},\ldots,v_n$. The \emph{value} of the potential~$p^\T_i$ is $\w(p^\T_i):=\sum_{l=1}^{\Delta}p^\T_i[l]$. 
\end{definition}
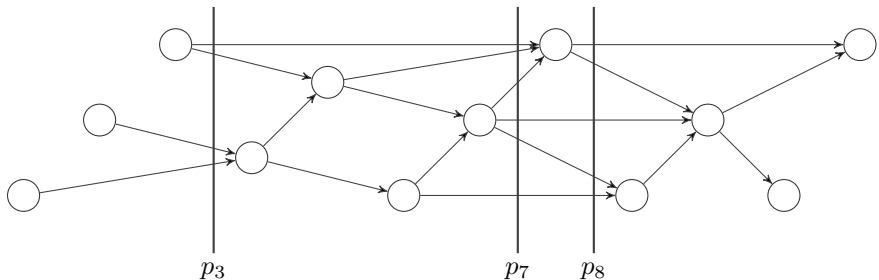
\begin{figure}[t]
	\begin{center}
		\def\layersep{1.5cm}
		\def\layersepp{1cm}
		\begin{tikzpicture}[draw=black!75, scale=1,->,>=stealth']
			\tikzstyle{vertex}=[circle,draw=black!80,minimum size=12pt,inner sep=0pt]

			\node[vertex] (V1) at (0,0) {};
			\node[vertex] (V2) at (1,1) {};
			\node[vertex] (V3) at (2,2) {};
			\node[vertex] (V4) at (3,0.5) {};
			\node[vertex] (V5) at (4,1.5) {};
			\node[vertex] (V6) at (5,0) {};
			\node[vertex] (V7) at (6,1) {};
			\node[vertex] (V8) at (7,2) {};
			\node[vertex] (V9) at (8,0) {};
			\node[vertex] (V10) at (9,1) {};
			\node[vertex] (V11) at (10,0) {};
			\node[vertex] (V12) at (11,2) {};

			\path (V1) edge (V4);
			\path (V2) edge (V4);
			\path (V3) edge (V5);
			\path (V3) edge (V8);
			\path (V4) edge (V5);
			\path (V4) edge (V6);
			\path (V5) edge (V7);
			\path (V5) edge (V8);
			\path (V6) edge (V7);
			\path (V6) edge (V9);
			\path (V7) edge (V8);
			\path (V7) edge (V9);
			\path (V7) edge (V10);
			\path (V8) edge (V10);
			\path (V8) edge (V12);
			\path (V9) edge (V10);
			\path (V10) edge (V11);
			\path (V10) edge (V12);

			\node[] (P1) at ( 2.5, -1) {$p_3$};
			\path (2.5,  2.5) edge[-,thick] (P1);
			\node[] (P2) at ( 6.5, -1) {$p_7$};
			\path (6.5,  2.5) edge[-,thick] (P2);
			\node[] (P3) at ( 7.5, -1) {$p_8$};
			\path (7.5,  2.5) edge[-,thick] (P3);
		\end{tikzpicture}
	\caption{A \TopOr for the example degree sequence~$\S =\left\{ \binom{0}{1},\binom{0}{1},\binom{0}{2},\binom{2}{2},\binom{2}{2},\binom{1}{2},\binom{2}{3},\binom{3}{2},\binom{2}{1},\binom{3}{2},\binom{2}{0},\binom{1}{0}\right \}$. The highlighted potentials are as follows:~$p_3=(3,1)^T$, $p_7 = (4,1,1)^T$, and $p_8 = (3,2)^T$.
	}
	\label{fig:potentialExample}
	\end{center}
\end{figure}
See \autoref{fig:potentialExample} for an example of the definition. 
If the \tOrd~$\T$ is clear from the context, then we write~$p$ instead of~$p^\T$.
Observe that, for any potential~$p_i \in \N^{\Delta}$, it holds that~$p_i[j] \ge p_i[j+1]$ for all~$1 \le j < \Delta$.
We denote with $0^\Delta$ the potential of value zero.

\paragraph{Algorithm Outline.}
Our algorithm consists of two parts. First, if the degree sequence of a \DAGR instance admits a dag realization where at any position the value of the potential is at least~$\pot$, then we will find such a ``high-potential'' realization with the algorithm that is described in \autoref{sec:high-potential-sequences}. Otherwise, by exploiting the fact that the value of all potentials is upper-bounded, we will find a ``low potential'' realization with the algorithm described in \autoref{sec:low potential-sequences}.

\subsection{General Terms and Observations}
In this section we introduce some general notations and observations that will be used in the algorithms to find high potential as well as low potential realizations.\medskip

\noindent \emph{Notation:} For a \tOrd $\T=v_1,\ldots,v_n$ and two indices $1\le i\le j\le n$, set $\T[i,j]:=v_i,v_{i+1},\ldots,v_j$. The set $\{v_i,\ldots,v_{j}\}$ is also denoted by $\phi[i,j]$.

\begin{definition}\label{def:type} Let $\S=\left \{ \binom{a_1}{b_1}, \ldots, \binom{a_n}{b_n}\right \}$ be a degree sequence. Two tuples $\binom{a_i}{b_i}$ and $\binom{a_j}{b_j}$ are of the same \emph{type} if $a_i=a_j$ and $b_i=b_j$.
Furthermore, a type $\binom{a_i}{b_i}$ is a \emph{good type} if $a_i\le b_i$ and otherwise it is a \emph{bad type}.
\end{definition}
Note that there are at most $(\Delta+1)^2$ different types.

\paragraph{Well-connected dags.}

\citet{BM11a} already observed that, given a degree sequence~$\mathcal{S}=\left\{\binom{a_1}{b_1},\ldots,\binom{a_n}{b_n}\right\}$, one can check in polynomial time whether~$\S$ is realizable by a dag with a corresponding \tOrd $v_1,\ldots,v_n$ where $d^-(v_i)=a_i$ and $d^+(v_i)=b_i$. This implies that it is sufficient to compute the correct ordering of the elements in~$\S$ as they appear in a \tOrd of a realizing dag. To prove this, the main observation is that for any \tOrd one can construct at least one corresponding dag by \emph{well-connecting} consecutive vertices. 

\begin{definition}\label{def:well-connected}
	Let~$D$ be a dag with a corresponding \tOrd $\phi=v_1,\ldots,v_n$. The \emph{remaining outdegree} at position $j$ of vertex~$v_i$, $1\le i \le j \le n$, is the number of $v_i$'s neighbors in the subsequence~$\phi[j,n]$.
	Furthermore, $D$ is \emph{well-connected} if for all vertices~$v_i \in \phi$ it holds that~$v_i$ is connected to the~$d^-(v_i)$ vertices in~$\phi[1,i-1]$ that have the highest remaining outdegree at position~$i-1$.
\end{definition}
As a consequence of \autoref{def:well-connected} we show that in a well-connected dag the potential at position~$i$ can be easily determined from that at position~$i-1$.

\begin{lemma} \label{lem:ordering-defines-potential}
	Let $\phi=v_1,\ldots,v_n$ be a \tOrd. Then, there is a well-connected dag~$D$ such that~$\phi$ is also a \tOrd for~$D$. Furthermore, for all $1<i\le n$ and $1\le j\le \Delta$ it holds that
\begin{align*}
   \begin{split}
		p_{i}[j] =& 
		\begin{cases}
			p_{i-1}[j] & \text{if } j<\Delta \wedge p_{i-1}[j+1]\ge d^-(v_i),\\
			p_{i-1}[j]-(d^-(v_i)-p_{i-1}[j+1]) & \text{if } j<\Delta \wedge p_{i-1}[j]\ge d^-(v_i), \\
			p_{i-1}[j+1] & \text{if } j<\Delta \wedge p_{i-1}[j]<d^-(v_i),\\
%
%
			\max\{0, p_{i-1}[j] - d^-(v_{i})\} & \text{if } j = \Delta. \\
		\end{cases}
\\
 +&
\begin{cases}
	1 & d^+(v_i)\ge j,\\
	0 & \text{otherwise}
\end{cases}
   \end{split}
\end{align*}

\end{lemma}

\begin{proof}
	Consider a \tOrd $\phi=v_1,\ldots,v_n$ for a dag~$D=(V,A)$ and suppose that the vertex~$v_i$ is not well-connected in~$D$. 
	The vertex~$v_i$ needs $d^-(v_i)$~ingoing arcs from the vertices in~$\phi[1,i-1]$ and has~$d^+(v_i)$ outgoing arcs to the vertices in~$\phi[i+1,n]$.
	Since~$v_i$ is not well-connected, there exist two vertices~$v_h, v_l \in \phi[1,i-1]$ such that~$v_l$ has lower remaining outdegree at position~$i-1$ than~$v_h$ and~$(v_l,v_i) \in A$ and~$(v_h,v_i) \notin A$.
	Moreover, since~$v_h$ has higher remaining outdegree at position~$i-1$ than~$v_l$ and $(v_h,v_i)\notin A$, there exists a vertex~$u \in \phi[i+1,n]$ such that~$(v_l,u) \notin A$ and~$(v_h,u) \in A$.
	Thus, $D' = (V,A')$ with $A' := (A \backslash \{(v_l,v_i), (v_h,u)\}) \cup \{(v_h,v_i), (v_l,u)\}$ is also a dag such that $\phi$ corresponds to~$D'$. By iteratively performing this operation we obtain a dag in which all vertices are well-connected.

	Now, consider the well-connected dag~$D$ with corresponding \tOrd $\phi=v_1,\ldots,v_n$. Then, the potential~$p_i$ computes from~$p_{i-1}$ as follows:
	The~$d^-(v_i)$ ingoing arcs to~$v_i$ decrease the outdegree of the $d^-(v_i)$~vertices with highest remaining outdegree at position~$i-1$ by one.
	Additionally, the vertex~$v_i$ has outdegree~$d^+(v_i)$ and, thus, all $p_i[j]$ with~$d^+(v_i)\ge j$ are increased by one.
\end{proof}

\paragraph{Cut-out subsequences.} The following lemma shows that if in a \tOrd $\phi[1,n]$ there are two indices~$1\le i<j\le n$ with equal potential, then we can cut out $\phi[i+1,j]$ resulting in a \tOrd~$\phi[1,i][j+1,n]$. We later show that we can reinsert $\phi[i+1,j]$ at any position that fulfills some reasonable conditions. This is the main operation that we perform in order to ``restructure'' a \tOrd such that we can exploit the resulting regular structure in our algorithms.

\begin{lemma}\label{lem:cut-out-neutral-block}
Let $\T=v_1,\ldots,v_n$ be a \TopOr for the degree sequence $\S=\left \{ \binom{a_1}{b_1}, \ldots, \binom{a_n}{b_n}\right \}$. If there are two indices $1\le i< j\le n$ such that $p^\T_i=p^\T_j$, then the sequence $\T'=\T[1,i] \T[j+1,n]$ is a \TopOr for the degree sequence that results from~$\S$ by deleting the degrees of the vertices in~$\T[i+1,j]$. Moreover, the potential $p^{\T'}_{i+l}$ is equal to~$p^\T_{j+l}$ for all~$1 \le l \le n-j$.
\end{lemma}
\begin{proof}
  Let $\S = \binom{a_1}{b_1}, \ldots, \binom{a_n}{a_n}$ be a degree sequence and let $1\le i<j\le n$ be two indices such that in a \TopOr $\T=v_1,\ldots,v_n$  it holds that $p^\T_i=p^\T_j$. Denoting by $\S'$ the degree sequence that results from~$\S$ by deleting all the degrees of the vertices in~$\T[i+1,j]$, we show that $\T':=\T[1,i]\T[j+1,n]$ is a \TopOr for~$\S'$.

  Let $V^{1-i}$ be all vertices in $\T[1,i]$ that have at least one neighbor in $\T[i+1,n]$ and, correspondingly, let $V^{1-j}$ be the vertices in $\T[1,j]$ that have at least one neighbor in $\T[j+1,n]$. By the definition of a potential, for all $1\le l\le \Delta$, the number of vertices in~$\T[1,i]$ that have exactly~$l$ neighbors in $\T[i+1,n]$ is equal to the number of vertices in $\T[1,j]$ that have exactly~$l$ neighbors in $\T[j+1,n]$. Thus, there is a bijection $f:V^{1-j}\rightarrow V^{1-i}$ such that for all~$v\in V^{1-j}$ it holds that vertex~$f(v)$ has the same number of neighbors in $\T[i+1,n]$ as~$v$ in $\T[j+1,n]$. Thus, deleting in the dag that corresponds to $\T$ the vertices in $\T[i+1,j]$ and exchanging every arc from a vertex $v\in \T[1,j]$ to a vertex $u\in \T[j+1,n]$ by $(f(v),u)$ results in a dag that is a realization for~$\S'$. Note that the vertex at position~$i+1$ in~$\T'$ is the same as the vertex at position~$j+1$ in~$\T$. Since $f$~is a bijection it is clear that the potential at position $i+l$ in~$\T'$ for~$1 \le l \le n-j$ is equal to the potential at position~$j+l$ in~$\T$.
\end{proof}
\autoref{lem:cut-out-neutral-block} shows that from a \tOrd $\phi$ we can cut out a subsequence~$\phi[i+1,j]$ whenever $p^{\phi}_i=p^{\phi}_j$. Informally speaking, we shall show that the subsequence $\phi[i+1,j]$ can be inserted into any \tOrd~$\phi'$ at position~$b$ whenever there is ``enough potential'' from the left part~$\phi'[1,b]$ to satisfy the indegrees of~$\phi[i,j]$. Then, the structure of~$\phi[i,j]$ guarantees that the ``remaining potential'' of $\phi'[1,b]\phi[i+1,j]$ is sufficient to satisfy the indegree of $\phi'[b+1,n]$ and thus $\phi'[1,b]\phi[i+1,j]\phi'[b+1,n]$ is a \tOrd. We need the following definition to formalize the conditions that $\phi[i+1,j]$ has to fulfill.

\begin{definition}\label{def:partial-top-order}
      Let $\S=\left \{ \binom{a_1}{b_1}, \ldots, \binom{a_n}{b_n}\right \}$ be a degree sequence and let $p^s,p^t\in \N^\Delta$ be two vectors.
      Furthermore, let~$P^s$ be a degree sequence with maximum degree~$\Delta$ consisting of~$p^s[1]$ elements such that for each $1\le l \le \Delta$ there are exactly~$p^s[l]$ elements $\binom{0}{b}$ where $b\ge l$.
      Correspondingly, let $P^t$ be a degree sequence consisting of $\w(p^t)$ entries, all of the form $\binom{1}{0}$.
      Let $\phi$ be a \TopOr for $\S \uplus P^s\uplus P^t$ where the vertices whose degrees correspond to~$P^s$ ($P^t$) are the first (last) vertices. If the potential at position~$p^s[1]+n$ is exactly~$p^t$, then $\phi[p^s[1]+1,p^s[1]+n]$ is a \emph{\pTopOr{}} for~$\S$ with input potential~$p^s$ and output potential~$p^t$.
\end{definition}
Note that in \autoref{def:partial-top-order} for the \TopOr $\phi$ for $\S\uplus P^s\uplus P^t$ it holds that the potential at position~$p^s[1]$ is $p^s$, and by definition at position $p^s[1]+n$ it is $p^t$.
Furthermore, for a \TopOr $\phi=v_1,\ldots,v_n$, for all $1\le i<j\le n$ it holds that $\phi[i+1,j]$ is a \pTopOr with input potential~$p_{i}$ and output potential~$p_j$.

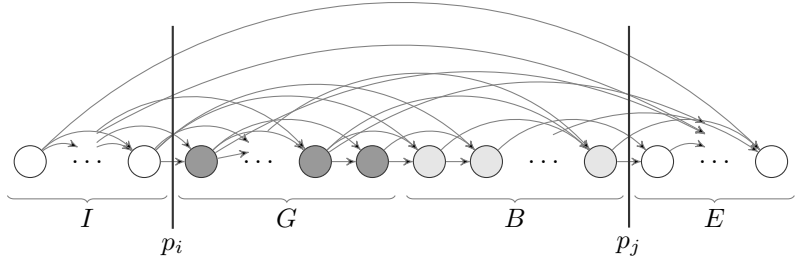
\begin{figure}[t]
	\begin{center}
		\begin{tikzpicture}[draw=black!50,decoration=brace, scale=1.5,->,>=stealth']
			\tikzstyle{vertex}=[circle,draw=black!80,minimum size=12pt,inner sep=0pt]

			\foreach \y in {1,...,3}{
				\node[] (s-\y) at (0.5,\y * 0.1) {};
				\node[] (m1-\y) at (2,\y * 0.1) {};
				\node[] (m2-\y) at (4.5,\y * 0.1) {};
				\node[] (e-\y) at (6,\y * 0.1) {};
			}

			\node[vertex] (S-S) at (0,0) {}; 
			\node[]  at (0.5,0) {$\ldots$};
			\node[vertex] (S-E) at (1,0) {}; 
			\node[vertex,fill=gray!80] (M-S) at (1.5,0) {}; 
			\node[]  at (2,0) {$\ldots$};
			\node[vertex,fill=gray!80] (M-M1) at (2.5,0) {}; 
			\node[vertex,fill=gray!80] (M-M2) at (3,0) {}; 
			\node[vertex,fill=gray!20] (M-M3) at (3.5,0) {}; 
			\node[vertex,fill=gray!20] (M-M4) at (4,0) {}; 
			\node[]  at (4.5,0) {$\ldots$};
			\node[vertex,fill=gray!20] (M-E) at (5,0) {}; 
			\node[vertex] (E-S) at (5.5,0) {}; 
			\node[]  at (6,0) {$\ldots$};
			\node[vertex] (E-E) at (6.5,0) {}; 

			\draw[-,decorate] (1.2,-.3) -- (-0.2,-.3);
			\node at (0.5,-0.5) {$I$};
			\draw[-,decorate] (3.2,-0.3) -- (1.3,-0.3);
			\node at (2.25,-0.5) {$G$};
			\draw[-,decorate] (5.2,-0.3) -- (3.3,-0.3);
			\node at (4.25,-0.5) {$B$};
			\draw[-,decorate] (6.7,-0.3) -- (5.3,-0.3);
			\node at (6,-0.5) {$E$};

 			\path (S-S) edge[bend left=45] (E-E);
 			\path (S-S) edge[bend left] (s-1);
			\path (s-1) edge[bend left] (S-E);
			\path (S-S) edge[bend left=45] (S-E);
			\path (S-E) edge[] (M-S);
			\path (S-E) edge[bend left=45] (m1-1);
			\path (S-E) edge[bend left=45] (M-M3);
			\path (S-E) edge[bend left=45] (M-M4);

			\path (M-S) edge[] (m1-1);
			\path (M-S) edge[bend left=45] (M-M1);
			\path (M-S) edge[bend left=45] (M-M2);
			\path (s-1) edge[bend left=45] (M-S);
			\path (s-2) edge[bend left=45] (M-M1);
			\path (M-M1) edge[] (M-M2);
			\path (M-M1) edge[bend left=45] (M-M3);
			\path (M-M1) edge[bend left=45] (M-E);
			\path (M-M2) edge[] (M-M3);
			\path (M-M2) edge[bend left=45] (M-M4);
			\path (M-M3) edge[] (M-M4);
			\path (M-M3) edge[bend left=45] (M-E);
			\path (M-M4) edge[bend left=45] (E-S);
			\path (m1-2) edge[bend left=45] (M-E);
			\path (M-E) edge[bend left=45] (E-E);
			\path (M-E) edge[] (E-S);

			\path (E-S) edge[bend left] (e-1);
			\path (M-S) edge[bend left] (e-2);
			\path (M-M1) edge[bend left] (e-3);
			\path (s-2) edge[bend left] (e-2);
			\path (m2-2) edge[bend left] (E-E);

			\node[] (P1) at ( 1.25, -0.75) {$p_i$};
			\path (1.25,  1.2) edge[-,thick,draw=black!80] (P1);

			\node[] (P2) at ( 5.25, -0.75) {$p_j$};
			\path (5.25,  1.2) edge[-,thick,draw=black!80] (P2);

		\end{tikzpicture} 
	\end{center}
	\caption{A realizing high potential topological ordering that corresponds to the pattern $I \circ G \circ B \circ E$. Thereby, $I$~is a subsequence of length at most $\Delta^{2\Delta}$ such that the first high potential occurs at position~$i$. Correspondingly, $j$~is the last position with high potential and~$E$  is a sequence of length at most $\Delta^{2\Delta}$. The sequence~$G$ (resp., $B$) only consists of good (bad) type vertices but is of arbitrary length. All high potential realizations can be reordered to fit into this pattern.}
	\label{fig:exampleHighPotential}
\end{figure}

\subsection{High Potential Sequences}\label{sec:high-potential-sequences}
In order to show that \DAGR is fixed-parameter tractable with respect to the parameter maximum degree~$\Delta$, in this subsection we show that if a realizable sequence admits a \TopOr where at some position the value of the potential is at least~$\Delta^2$, a so-called \emph{high potential \TopOr}, then there is also a \TopOr~$\phi$ that is of the following ``pattern'' (see \autoref{fig:exampleHighPotential} for an illustration): The ordering~$\phi$ can be partitioned into four sub-sequences $I \circ G \circ B \circ E$ (where~$\circ$ is the concatenation). The sequence~$I$ is an initializing sequence that ``establishes'' a potential of value at least~$\Delta^2$, a so-called \emph{high potential}. Correspondingly, at the end there is a sequence~$E$ that reduces the value of the potential from a value that is greater than~$\Delta^2$ to zero. Furthermore, $I$ and~$E$ are of length at most~$\Delta^{2\Delta}$ and thus can be guessed in $O((\Delta+1)^2)^{\Delta^{2\Delta}})=O(\Delta^{2\Delta^{2\Delta}})$ time. The subsequence~$G$, which is of arbitrary length, only consists of good types and, correspondingly,~$B$ is of arbitrary length but only consists of bad types in arbitrary order.

Our strategy to prove that there is a high potential \TopOr with the pattern $I \circ G \circ B \circ E$ is as follows. Let $\phi=v_1,\ldots,v_n$ be an arbitrary high potential \TopOr and let~$i$ be the minimum position with high potential and, symmetrically, let~$j$ be the maximum position with high potential. In the first part of this subsection (see \autoref{prop:short-begin-end-high-potential}), we show that we can assume that $i\le \Delta^{2\Delta}$ and $j\ge n-\Delta^{2\Delta}$. Towards this the main argument is that if $i>\Delta^{2\Delta}$, since there are $O(\Delta^{2\Delta})$ potentials with value less than~$\Delta^2$, there have to be two positions $1\le l_1<l_2<i$ with $p_{l_1}=p_{l_2}$. Then, by \autoref{lem:cut-out-neutral-block}, we can cut out $\phi[l_1+1,l_2]$ from $\phi$ and we will show (see \autoref{lem:introduce-neutral-block-into-high}) that we can reinsert it right behind~$i$, resulting in a \TopOr $\phi[1,l_1]\phi[l_2+1,i]\phi[l_1+1,l_2]\phi[i+1,n]$. By iteratively applying this operation, we end up with a \TopOr where the minimum position with high potential is at most $\Delta^{2\Delta}$.  A symmetric argument holds for the maximum position~$j$ with high potential.

In the second part 
we show that we can arbitrarily sort the vertices in~$\phi[i+1,j]$  under the constraint that at first vertices of good type occur in any order, and then they are followed by the bad type vertices (see \autoref{prop:sort-between-high-potential}). Altogether, this shows that in order to check whether there is a high potential \TopOr it is sufficient to branch into all possibilities to choose~$I$ and~$E$, insert the remaining vertices sorted by good and bad types between~$I$ and~$E$, and, finally, check whether this ordering is a \tOrd.

We first prove that if we have two \pTopOr{s}~$\phi_1$ and~$\phi_2$ where~$\phi_1$ has input potential~$0^\Delta$ and $\phi_2$ has output potential~$0^\Delta$ and the output potential of~$\phi_1$ is at ``least as good'' as the input potential of $\phi_2$, then we can merge them to a \TopOr $\phi_1\phi_2$ while preserving the indegree and outdegree of all vertices. Before proving that, we define a partial order for potentials.

\begin{definition}\label{def:partial-order-potential}
For $p,p'\in \N^\Delta$, $p\ge p'$ if $\forall 1\le j\le \Delta: \sum_{i=1}^{j}p[i]\ge \sum_{i=1}^{j}p'[i]$.
\end{definition}
The intuition of \autoref{def:partial-order-potential} is that a potential~$p$ is at least as good as a potential~$p'$ if subsequent vertices that can be connected with potential~$p'$ can also be connected with potential~$p$. To gurantee that there are enough vertices of degree at least~$i$, it either has to hold that $p[i]\ge p'[i]$ or there is a sufficiently large ``overhang'' of vertices with degree less than~$i$ that are not necessary to gurantee the existence of vertices with degree less than~$i$. Formally, $\sum_{j=1}^{i-1}p[j]-\sum_{j=1}^{i-1} p'[j]\ge p'[i]-p[i]$.

\begin{lemma}\label{lem:opt-subsequence}
 Let $\phi=v_1,\ldots,v_n$ be a \TopOr for a degree sequence~$\S$, and let $1\le i\le n$ be an arbitrary position. For any \pTopOr~$\phi'$ for a degree sequence~$\S'$ with input potential~$0^\Delta$ and output potential $p\in \N^\Delta$ with $\w(p)=\w(p^\phi_i)$ and $p\ge p^\phi_i$, the sequence $\phi'\phi[i+1,n]$ is a \TopOr for~$\S'\uplus  \left\{\binom{d^-(v_{i+1})}{d^+(v_{i+1})}, \ldots, \binom{d^-(v_{n})}{d^+(v_{n})}\right\}$.
\end{lemma}
\begin{proof}
    Let $\phi=v_1,\ldots,v_n$ be a \TopOr for a degree sequence~$\S$, and let $1\le i\le n$ be an arbitrary position. Furthermore, let $\phi'$ be a \pTopOr for a degree sequence~$\S'$ with input potential~$0^\Delta$ and output potential $p\in \N^\Delta$ with $\w(p)=\w(p^\phi_i)$ and $p\ge p^{\phi}_i$ (see \autoref{def:partial-order-potential}). By \autoref{def:partial-top-order}, there are two degree sequences~$P^s$ and $P^t$ such that there is a \TopOr~$\phi'_{s,t}$ for $P^s\uplus P^t\uplus \S'$ where the vertices that correspond to $P^s$ ($P^t$) are the first (last) vertices.

    We show that the sequence $\phi'\phi[i+1,n]$ is a \TopOr for~$\S'\uplus \left\{ \binom{d^-(v_{i+1})}{d^+(v_{i+1})}, \ldots, \binom{d^-(v_{n})}{d^+(v_{n})}\right\}$.
    Therefore, we construct a dag~$D$ that corresponds to $\phi'\phi[i+1,n]$ and thus is a realization. We first copy all arcs between two vertices in $\phi[i+1,n]$ that are present in the dag for~$\phi$ and, correspondingly, all arcs between two vertices in~$\phi'$ that are present in the dag for~$\phi'_{s,t}$. Now, the potential in $\phi'$ at position~$|\phi'|$ is~$p$. By the condition of \autoref{lem:opt-subsequence}, it holds that $p\ge p^{\phi}_i$ and $\w(p)=\w(p^{\phi}_i)$.

    Now, we show how to connect the $n-i$ vertices in~$\phi[i+1,n]$ to their ancestors in~$\phi'$. Specifically, for $1\le r\le n-i$, let $v_r$ be the $r^\text{th}$ vertex in~$\phi[i+1,n]$. We show how to connect the vertex~$v_r$ to its ancestors such that the potential at position $|\phi'|+r$ in $\phi'\phi[i+1,n]$, denoted by~$p^r$, is greater than $p^{\phi}_{i+r}$.
    To this end, we use induction, meaning that we assume that the potential of~$p^{r-1}$ is greater than~$p^{\phi}_{i-1+r}$. (Clearly, at the beginning for $r=1$, the direct ancestor of~$v_r$ is the last vertex of~$\phi'$, and thus we set $p^0=p$, implying that $p^0=p\ge p^{\phi}_i$.) First, since $p^{r-1}[1]\ge p^{\phi}_{i-1+r}[1]$ it follows that we can well-connect~$v_r$ to its ancestors. 
	Note that for all~$j$, $1 \le j \le n-1$, it holds that~$\w(p^j) = \w(p^\phi_{i+j})$.
	We next prove that for the resulting potential~$p^r$ it holds that $p^r\ge p^\phi_{i+r}$. This completes our argumentation. 

	\newcommand{\p}{f}
	\newcommand{\pd}{\p^-}
	\newcommand{\e}{e}
	\newcommand{\ed}{\e^-}

	To this end, denoting by~$c\in N^\Delta$ the vector that has ones in the first $d^+(v_r)$ rows and the remaining entries are zero, by the definition of potentials it is clear that $p^r\ge p^\phi_{i+r}\Leftrightarrow p^r-c\ge p^\phi_{i+r}-c$.
    For the sake of readability we substitute as follows $\p=p^{r-1},\pd=p^r-c$, and $\e=p^{\phi}_{i+r-1}, \ed=p^{\phi}_{i+r}-c$ and we shall show that $\pd\ge \ed$. Note that from $\w(\p)=\w(\e)$ it follows that $\w(\p)-d^-(v_r)=\w(\pd)=\w(\ed)$. Towards a contradiction assume that there is a position $1\le l<\Delta$ such that
	\begin{equation}\label{eq:assumption}
		\sum_{j=1}^l \ed[j]-\pd[j]>0.
	\end{equation}
	Since $\p\ge \e$, it follows that 
	\begin{equation}\label{eq:diff}
		\sum_{j=1}^{l}\p[j]-\pd[j]>\sum_{j=1}^{l}\e[j]-\ed[j]
	\end{equation}
	and from this together with $\w(\pd)=\w(\ed)$ we can infer that
	\begin{equation*}
		\sum_{j=l+1}^{\Delta} \p[j]-\pd[j]<\sum_{j=l+1}^{\Delta} \e[j]-\ed[j].
	\end{equation*}
	By \autoref{lem:ordering-defines-potential} it follows that
	\begin{equation*}
		\sum_{j=l+1}^{\Delta} \p[j]-\pd[j]= \p[l+1] \text{ and } \sum_{j=l+1}^{\Delta} \e[j]-\ed[j]\le \e[l+1]
	\end{equation*}
From that and since $\w(\p)=\w(\e)=\w(\pd)+d^-(v_r)=\w(\ed)+d^-(v_r)$, for Inequality~\eqref{eq:diff} it follows from \autoref{lem:ordering-defines-potential} that 
\begin{align}
 \left(\sum_{j=1}^{l}\p[j]-\pd[j]\right)-\left(\sum_{j=1}^{l}\e[j]-\ed[j]\right)& = \notag \\
 \label{eq:dist} \left(\sum_{j=1}^{l}\p[j]-\e[j]\right)+\left(\sum_{j=1}^{l}\ed[j]-\pd[j]\right)& \le \e[l+1]-\p[l+1]
\end{align}
Since from $\p\ge \e$ it follows that $\sum_{j=1}^l \p[j]-\e[j]\ge \e[l+1]-\p[l+1]$ and this implies together with Inequality \eqref{eq:dist} that 
\begin{equation*}
 \sum_{j=1}^{l} \ed[j]-\pd[j]\le 0,
\end{equation*}
causing a contradiction to Inequality \eqref{eq:assumption}.
\end{proof}
\autoref{lem:opt-subsequence} shows that we can ``merge'' two \pTopOr{s} $\phi_1$ and~$\phi_2$ to~$\phi_1\phi_2$, if for the output potential~$p^{\phi_1}_o$ of~$\phi_1$ and the input potential~$p^{\phi_2}_i$ it holds that $p^{\phi_1}_o\ge p^{\phi_2}_i$ and $\w(p^{\phi_1}_0)=\w(p^{\phi_2}_i)$. The next lemma shows that the condition $p^{\phi_1}_o\ge p^{\phi_2}_i$ is not necessary in case of high potentials, that is, $\w(p^{\phi_1}_0)=\w(p^{\phi_2}_i)\ge \Delta^2$.
Before that, we need the following observation showing that for a fixed value there is a potential that is less than all others.
\begin{observation}\label{obs:worse-potential}
	For a fixed positive integer~$x$ let $p\in \N^\Delta$ be the potential with $$ p[j] = \begin{cases} \left\lceil\frac{x}{\Delta}\right\rceil, & \text{if } j \le x \text{ modulo } \Delta \\ \left\lfloor\frac{x}{\Delta}\right\rfloor, & \text{otherwise} \end{cases}$$
for all $1\le j\le \Delta$. Then, for all potentials $p'\in \N^\Delta$ with $\w(p')=\w(p)$ it holds that $p'\ge p$.
\end{observation}
\begin{proof}
 Let~$p\in \N^\Delta$ be the potential as defined in \autoref{obs:worse-potential} and let $p'\in \N^\Delta$ be a potential with $\w(p')=\w(p)$. Clearly, by definition it holds that $\w(p)=x$. Towards a contradiction assume that $p'<p$. Hence, there is a position $1\le j\le \Delta$ with $\sum_{l=1}^j p[l]>\sum_{l=1}^j p'[l]$.
 From this it follows that there is a position $1\le t\le j$ such that $p[t]>p'[t]$ and since $p[t]\le \lceil\nicefrac{x}{\Delta}\rceil$ it follows that $p'[t]\le \lfloor \nicefrac{x}{\Delta}\rfloor$. Recall that, by definition, for any potential it holds that $p[l_1]\ge p[l_2]$ for all $1\le l_1\le l_2\le \Delta$. Thus, from $p'[t]\le \lfloor\nicefrac{x}{\Delta}\rfloor$ it follows that $\sum_{l=j+1}^\Delta p'[l]\le (\Delta-j)\lfloor\nicefrac{x}{\Delta}\rfloor \le  \sum_{l=j+1}^\Delta p[l]$. Together with $\sum_{l=1}^j p[l]>\sum_{l=1}^j p'[l]$ this implies a contradiction to $\w(p)=\w(p')$.
\end{proof}

\begin{lemma}\label{lem:high-pot}
	Let $\phi=v_1,\ldots,v_n$ be a \TopOr for a degree sequence~$\S$ and let $\w(p^\T_i) \ge \Delta^2$ for some $1\le i\le n$. Then, for any \pTopOr $\phi'$  for a degree sequence~$\S'$ with input potential~$0^\Delta$ and output potential~$p$ with $\w(p)=\w(p^\T_i)$, the sequence~$\phi'\phi[i+1,n]$ is a \TopOr for~$\S'\uplus \left\{ \binom{d^-(v_{i+1})}{d^+(v_{i+1})}, \ldots, \binom{d^-(v_{n})}{d^+(v_{n})}\right\}$.
\end{lemma}
\begin{proof}
  Let $\T=v_1,\ldots,v_n$ be a \TopOr for a degree sequence~$\S$ and let $1\le i\le n$ be a position with $\w(p^\T_i)\ge \Delta^2$. Furthermore, let~$\T'$ be a \pTopOr with input potential~$0^\Delta$ and output potential~$p$ with $\w(p)=\w(p^\T_i)=x$. We shall show that $\T'\T[i+1,n]$ is a \TopOr for $\S'\uplus \left\{ \binom{d^-(v_{i+1})}{d^+(v_{i+1})}, \ldots, \binom{d^-(v_{n})}{d^+(v_{n})}\right\}$. We prove \autoref{lem:high-pot} in case of $$ p[j] = \begin{cases} \left\lceil\frac{x}{\Delta}\right\rceil, & \text{if } j \le x \text{ modulo } \Delta \\ \left\lfloor\frac{x}{\Delta}\right\rfloor, & \text{otherwise} \end{cases}$$
  for all $1\le j\le \Delta$. Then, \autoref{obs:worse-potential} and \autoref{lem:opt-subsequence} imply its correctness in the general case.

  In the following, we describe how to construct a dag that corresponds to $\T'\T[i+1,n]$. We first add all arcs between two vertices in~$\phi[1,i]$ that are present in a dag for~$\T$. Correspondingly, we add all arcs between two vertices that are present in a dag for~$\T'$. Observe that it now only remains to add the arcs from a vertex in~$\T'$ to~$\T[i+1,n]$. For a more convenient construction of these arcs, assume that there are no arcs between two vertices in~$\T[i+1,n]$. (Clearly, because in the following we only add arcs having one endpoint in~$\T'$ and the other in~$\T[i+1,n]$, arcs between two vertices in~$\T[i+1,n]$ can be removed and, correspondingly, the degrees in~$\S$ can be adjusted. Afterwards, these arcs can be reinserted.)

  We next prove that it is possible to stepwise well-connect the vertices~$v_{i+r}$ for $r=1$ to $n-i$ to the vertices in $\T'$. By the assumption that there is no arc between vertices in~$\T[i+1,n]$, it holds that $\sum_{v\in \T[i+1,n]}d^-(v)=\w(p)=\w(p^\T_i)$. More specifically, denoting the potential at position~$|\T'|+r$ in $\T'\T[i+1,n]$ by~$p^r$, it is clear that $\w(p^\T_{i+r})=\w(p^r)$. 
  Now, towards a contradiction, assume that for some $1\le r \le n-i$ it is not possible to well-connect the vertex $v_{i+r}$ to $\T'\T[i+1,i+r-1]$. 
Clearly, since we cannot connect~$v_{i+r}$, it holds that $p^{r-1}[1]<d^-(v_{i+r})\le p^\T_{i+r-1}[1]$. 
 
  For $0\le j\le r-1$ and for a vertex~$v\in \T'$ consider the remaining outdegree~$d^+_j(v)$, that is, the outdegree of~$v$ minus the number arcs from~$v$ to any vertex in $\T'\T[i+1,i+j]$. Since we always choose the vertices with highest remaining outdegree from~$\T'$ to connect a vertex in~$\T[i+1,n]$ it follows that 
\begin{equation}\label{eq:gap-is-gett-smaller}
\max\{1,|d^+_j(v_1)-d^+_j(v_2)|\}\ge |d^+_{j+1}(v_1)-d^+_{j+1}(v_2)| 
\end{equation}
for all $0\le j<r-1$ and $v_1,v_2\in \T'$.
  Moreover, recall that, since $x\ge \Delta^2$, $p[\Delta]=\lfloor \nicefrac{x}{\Delta}\rfloor \ge \Delta$ and, thus, there are at least $\Delta$ vertices $L\subseteq \T'$ with $d^+_0(v)=\lfloor \nicefrac{x}{\Delta}\rfloor$. Because of Inequality~\eqref{eq:gap-is-gett-smaller} it follows that $\forall v_1,v_2\in L: |d^+_{r-1}(v_1)-d^+_{r-1}(v_2)|\le 1$. Moreover, since $p^{r-1}[1]<d^-(v_{i+r})\le \Delta$, there is a vertex $v\in L$ such that $d^+_{r-1}(v)=0$ and thus $d^+_{r-1}(w)\le 1$ for all $w\in L$. 

 It remains to consider the vertex (by definition there can be at most one) $u\in \T'\setminus L$. Because the remaining outdegree can only decrease by one in each step and $d^+_{r-1}(v)=0$, there is a position $1\le l\le r-1$ with $d^+_{l}(v)=d^+_{l}(u)$ and, thus, by Inequality~\eqref{eq:gap-is-gett-smaller} $d^+_{r-1}(u)\le 1$. Thus, $d^+_{r-1}(v)\le 1$ for all $v\in \T'$. Hence, $\w(p^{r-1})=p^{r-1}[1]$ 
  and $\w(p^{r-1}=\w(p^\T_{i+r-1})$ imply a contradiction to $p^{r-1}[1]<d^-(v_{i+r})\le p^\T_{i+r-1}[1]$. 
\end{proof}
While \autoref{lem:cut-out-neutral-block} shows that we can cut out a \pTopOr with equal input and output potential, the following lemma shows that we can reinsert it right behind a high potential in any \TopOr.

\begin{lemma}\label{lem:introduce-neutral-block-into-high}
 Let $\phi$ be a \TopOr for a degree sequence~$\S$. Furthermore, let $\phi'$ be a \pTopOr with equal input and output potential~$p$ for a degree sequence~$\S'$. Then, for any position~$1\le i\le |\phi|$ with $\w(p^\T_i)\ge \Delta^2$ and $\w(p^\T_i)\ge \w(p)$ it holds that $\phi[1,i]\phi'\phi[i+1,n]$ is a \TopOr for $\S\uplus \S'$.
\end{lemma}
\begin{proof}
For a degree sequence~$\S$ let $\phi$ be a \TopOr and let $\phi'$ be  a \pTopOr with input and output potential~$p$ for a degree sequence~$\S'$. Furthermore, let $1\le i\le |\phi|$ be a position with $\w(p^\T_i)\ge \Delta^2$ and $\w(p^\T_i)\ge \w(p)$. We prove that $\phi[1,i]\phi'\phi[i+1,n]$ is a \TopOr for $\S\uplus \S'$.

We first show that $\phi[1,i]\phi'$ is a \pTopOr with input potential $0^\Delta$ and output potential~$p$ where $\w(p)=\w(p^\T_i)$. Then, from \autoref{lem:high-pot} it follows that $\phi[1,i]\phi'\phi[i+1,|\phi|]$ is a \TopOr for $\S\uplus \S'$.

By \autoref{def:partial-top-order} there are two degree sequences~$P^s$ and~$P^t$ such that there is a \TopOr $\phi_{s,t}=\phi^s\phi'\phi^t$ for $P^s\uplus P^t\uplus \S'$ such that~$\phi^s$ contains the vertices that match to~$P^s$ and~$\phi^t$ that match to~$P^t$. 
Now, in case of $\w(p^\T_i)>\w(p)$ we extend~$P^s$ by $\w(p_i)-\w(p)$ elements of type $\binom{1}{0}$ and $P^t$ by the same number of $\binom{0}{1}$ type elements. This shows that $\phi'$ is a \pTopOr with input potential~$0^\Delta$ and output potential~$p$ with $\w(p)=\w(p^\T_i)$.

From this together with \autoref{lem:high-pot} it follows that $\phi[1,i]\phi'\phi^t$ is a \TopOr. Since, by our assumption $\phi'$ is a \pTopOr with input and output potential~$p$ it follows that $\sum_{v\in \phi'} d^-(v)=\sum_{v\in \phi'} d^+(v)$. From this, since the potential at position~$i$ in $\phi[1,i]\phi'\phi^t$ is~$p^\T_i$ it follows that $\phi[1,i]\phi'$ is a \pTopOr with input potential~$0^\Delta$ and output potential~$p'$ with $\w(p')=\w(p^\T_i)$. Thus, by \autoref{lem:high-pot} it follows that $\phi[1,i]\phi'\phi[i+1,n]$ is a \TopOr.
\end{proof}
 With \autoref{lem:introduce-neutral-block-into-high} we are able to bound the minimum and maximum position where a high potential occurs.
\begin{proposition}\label{prop:short-begin-end-high-potential}
	If a \DAGR instance admits a high-potential realization, then there is also a high potential \TopOr such that the minimum position with high potential is at most $\Delta^{2\Delta}$ and the maximum position with high potential is at least $n-\Delta^{2\Delta}$.
\end{proposition}
\begin{proof}
 Let $\phi$ be a high-potential \TopOr and let $1\le i\le n$ be the minimum position where $\w(p_i)\ge \Delta^2$. Consider the case where ~$i>\Delta^{2\Delta}$. Thus, for all $1\le l <i$ it holds that $\w(p_l)<\Delta^2$. However, there are less than $\Delta^{2\Delta}$ potentials with value less than~$\Delta^2$ and, thus, there are two indices $1\le l_1<l_2<i$ with $p_{l_1}=p_{l_2}$. By \autoref{lem:cut-out-neutral-block}, the sequence $\phi[1,l_1] \phi[l_2+1,n]$ is a \TopOr where the potential at position $i-(l_2-l_1)$ is~$p_i$. Moreover, by definition $\phi[l_1+1,l_2]$ is a \pTopOr with input and output potential~$p_{l_1}$ where $\w(p_{l_1})<\w(p_i)$. Thus, by \autoref{lem:introduce-neutral-block-into-high} it holds that $\phi[1,l_1]\phi[l_2+1,i]\phi[l_1+1,l_2]\phi[i+1,n]$ is a \TopOr. Moreover, in this \TopOr, since $\sum_{v\in \phi[l_1+1,l_2]}d^-(v)-d^+(v)=0$, the minimum position with high potential is $i-(l_2-l_1)$. Applying the same operation iteratively as long as there are two positions with equal potential before the first high-potential results in a \TopOr where the minimum position with high potential is at most~$\Delta^{2\Delta}$.

  Basically, the same argumentation can be applied for the maximum position~$j$ where a high potential occurs. In case of $j<n-\Delta^{2\Delta}$, there have to be two indices $j< l_1 < l_2 \le n$ where $p_{l_1}=p_{l_2}$. Then, by  \autoref{lem:cut-out-neutral-block} the sequence $\phi[1,l_1]\phi[l_2+1,n]$ is a \TopOr and $\phi[l_1+1,l_2]$ is a \pTopOr with input and output potential~$p_{l_1}$ with $\w(l_1) < \w(p_j)$. Thus, by \autoref{lem:introduce-neutral-block-into-high} the sequence $\phi[1,j]\phi[l_1+1,l_2]\phi[j+1,l_1]\phi[l_2+1,n]$  is a \TopOr where the maximum position with high potential is $j+(l_2-l_1)$. Again, by applying this operation iteratively we get a sequence where the maximum position with high potential is at least $n-\Delta^{2\Delta}$.
\end{proof}
%
Having shown that we can assume that for the minimum position~$i$ and the maximum position~$j$ with high potential it holds that $i\le \Delta^{2\Delta}$ and $j\ge n-\Delta^{2\Delta}$ we next prove that one can sort all vertices between~$i$ and~$j$ arbitrarily by good and bad types.

\begin{proposition}\label{prop:sort-between-high-potential}
	Let $\phi=v_1,\ldots,v_n$ be a high potential \TopOr for a degree sequence~$\S$ and let $1\le i<j\le n$ be two arbitrary indices such that $\w(p_i) \ge \Delta^2$ and~$\w(p_j) \ge \Delta^2$.
	Furthermore, let $\phi'[i+1,j]$ be a permutation of the vertices in~$\phi[i+1,j]$ such that there is a position~$0\le l\le j-i$ with the property that the first~$l$ vertices in~$\phi'[i+1,j]$ are of good type and all subsequent vertices are of bad type.
	Then, the sequence $\phi[1,i]\phi'[i+1,j]\phi[j+1,n]$ is a \TopOr for~$\S'$.
\end{proposition}
\begin{proof}
	Assume that there is a high potential \TopOr for a degree sequence~$\S$ with two indices $1\le i\le j\le n$ such that $\w(p_i)\ge \Delta^2$ and $\w(p_j)\ge \Delta^2$. 
	We prove that $\phi[1,i]\phi'[i+1,j]\phi[j+1,n]$ is a \TopOr for~$\S$ for any reordering~$\phi'[i+1,j]$ of the vertices in~$\phi[i+1,j]$ where first the vertices of good types are consecutive in any ordering and then are followed by the bad type vertices.

	To this end, by induction on~$l$ with $1\le l\le j-i$  we show that the sequence $\phi[1,i]\phi'[i+1,i+l]$ is a \pTopOr with input potential~$0^\Delta$ and output potential~$p^l$ with~$\w(p^l)\ge \Delta^2$. First, we well-connect the vertex $\phi'[i+l,i+l]$ to the \pTopOr $\phi[1,i]\phi'[i+1,i+l-1]$ to get $\phi[1,i]\phi'[i+1,i+l]$. This is always possible since the output potential~$p^{l-1}$ of $\phi[1,i]\phi'[i+1,i+l-1]$ is a high potential. It only remains to show that the value of the output potential~$p^l$ of $\phi[1,i]\phi'[i+1,i+l]$ is at least~$\Delta^2$. Towards a contradiction suppose that it is not. This implies
	\begin{equation}\label{eq:con-assum}
		\sum_{v\in \phi'[i+1,i+l]} d^-(v)-d^+(v)>\w(p_i)-\Delta^2.
	\end{equation}
	Clearly, the vertex $\phi'[i+l,i+l]$ has to be a bad type, otherwise \autoref{eq:con-assum} cannot be true.
	However, it holds that
\begin{equation*}
 \w(p_i)-\sum_{v\in \phi[i+1,j]}(d^-(v)-d^+(v))=\w(p_j)\ge \Delta^2
\end{equation*}
and  thus 
	\begin{equation}\label{eq:con-sum}
		\sum_{v\in \phi[i+1,j]} d^-(v)-d^+(v)\le \w(p_i)-\Delta^2.
	\end{equation}
	Since $\phi'[i+1,j]$ is sorted by good and bad types and $\phi'[i+l,i+l]$ is of bad type, all vertices in $\phi'[i+l,j]$ are bad type vertices. Thus, \autoref{eq:con-sum} yields a contradiction to \autoref{eq:con-assum}.
\end{proof}
\autoref{prop:short-begin-end-high-potential} and \autoref{prop:sort-between-high-potential} lead to the central contribution of this section.

\begin{theorem}\label{thm:highPotential}
	If a \DAGR instance admits a high potential \TopOr, then it can be solved in $O(\Delta^{4{\Delta^{2\Delta}}}\cdot n)$ time.
\end{theorem}
\begin{proof}
 If an instance of \DAGR admits a high potential \TopOr, then by \autoref{prop:short-begin-end-high-potential} there is also a high potential \TopOr in which the occurrence of the first high potential is at most at position $\Delta^{2\Delta}$ and the last occurrence of a high potential is at least at position $n-\Delta^{2\Delta}$. Recall that there are at most~$(\Delta+1)^2$ types of elements in the given degree sequence, and thus by exhaustive search we can find these two subsequences in time $O(\Delta^{4{\Delta^{2\Delta}}}\cdot n)$. \autoref{prop:sort-between-high-potential} shows that the remaining degrees can be arbitrarily inserted between them, as long as they are sorted by good and bad types.
\end{proof}

\subsection{Low Potential Sequences}\label{sec:low potential-sequences}

In this section, we will provide an algorithm that finds a low potential realization (if it exists) for a \DAGR instance.
That is, a realization such that in the corresponding \tOrd the value of all potentials is strictly less than~$\Delta^2$. 
See \autoref{fig:low potential-example} for an example of such a realization.

\begin{figure}[t]
	\begin{center}
	\def\layersep{1.5cm}
	\def\layersepp{1cm}
	\begin{tikzpicture}[shorten >=1pt,->,draw=black!50, node distance=\layersep]
		\tikzstyle{every pin edge}=[<-,shorten <=1pt]
		\tikzstyle{vertex}=[circle,draw=black!80,minimum size=12pt,inner sep=0pt]
		\tikzstyle{$x$-vertex}=[vertex, fill=white];
		\tikzstyle{$a$-vertex}=[vertex, fill=white];
		\tikzstyle{annot} = [text width=4em, text centered]
		\node[vertex] (S1) at (0,0) {};
		\node[vertex] (S2) at (0,-\layersep) {};

		\node[vertex] (I1) at (\layersepp,0) {};
		\node[vertex] (I2) at (2 * \layersepp,0) {};

		\node[vertex] (I3) at (3 * \layersepp,-\layersep) {};
		\node[vertex] (I4) at (4 * \layersepp,-\layersep) {};

		\path (S1) edge (I1);
		\path (S2) edge (I1);
		\path (I1) edge (I2);
		\path (S1) edge[bend left=45] (I2);
		\path (S2) edge (I2);

		\path (I2) edge (I3);
		\path (S2) edge (I3);
		\path (I3) edge (I4);
		\path (S2) edge[bend right] (I4);
		\path (I2) edge (I4);

		\foreach \y in {1,...,6}{
			\node[minimum size=20pt] (z-\y) at (5 * \layersepp, -\y / 2 + 1) {};
		}
		\foreach \y in {3,...,6}{
			\path (I4) edge (z-\y);
		}
		\path (I2) edge (z-1);
		\path (I2) edge (z-2);

		\node[] (dots) at (5 * \layersepp, - \layersep / 2) {$\cdots$};

		\foreach \y in {1,...,6}{
			\node[minimum size=20pt] (x\y) at (5 * \layersepp, -\y / 2 + 1) {};
		}

		\node[vertex] (I11) at (6 * \layersepp,0) {};
		\node[vertex] (I12) at (7 * \layersepp,0) {};

		\node[vertex] (I13) at (8 * \layersepp,-\layersep) {};
		\node[vertex] (I14) at (9 * \layersepp,-\layersep) {};

		\foreach \y in {1,...,4}{
			\node[vertex] (T\y) at (10 * \layersepp, \y * -\layersep / 2 + 3 *\layersep / 4) {};
		}

		\path (x2) edge (I11);
		\path (x3) edge (I11);
		\path (I11) edge (I12);
		\path (x1) edge[bend left=45] (I12);
		\path (x4) edge (I12);

		\path (I12) edge (I13);
		\path (x5) edge (I13);
		\path (I13) edge (I14);
		\path (x6) edge[bend right] (I14);
		\path (I12) edge (I14);

		\path (I12) edge (T1);
		\path (I12) edge (T2);
		\foreach \y in {1,...,4}{
			\path (I14) edge (T\y);
		}

	\end{tikzpicture}
	\end{center}
	\caption{Realization for the sequence~$\binom{0}{2},\binom{0}{4},\binom{2}{1},\binom{3}{4},\binom{2}{1},\binom{3}{4},\binom{2}{1},\binom{3}{4}, \ldots ,\binom{2}{1},\binom{3}{4}$, $\binom{2}{0}$, $\binom{2}{0}$, $\binom{1}{0}$, $\binom{1}{0}$. Since this sequence basically consists of only two different types (not regarding types with indegree or outdegree equal to zero), it is easy to check that the pictured low potential realization (the highest occurring value of a potential is six) is the only existing one. The main part of the \TopOr consists of a repetition of $\binom{2}{1},\binom{3}{4}$. }
	\label{fig:low potential-example}
\end{figure}
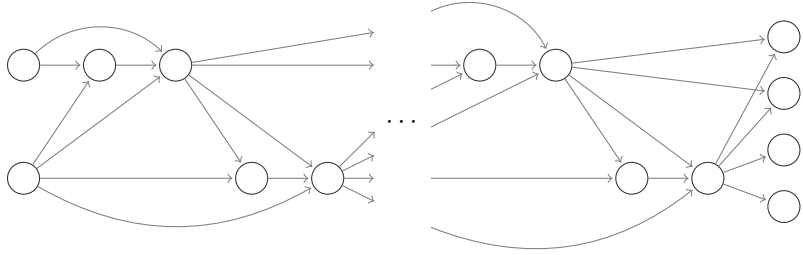
The crucial point to give an algorithm which solves such instances is that, besides some ``special gaps'' which can be handled afterwards, the length of a corresponding \TopOr can be upper bounded by a function~$f$ only depending on the maximal degree~$\Delta$.
Then, the algorithm, basically, consists of branching into all \TopOr{s} of length of at most~$f(\Delta)$ and, then, filling up the ``special gaps'' afterwards.

In the following we describe how to upper-bound the length. To this end, we introduce some notation.

\begin{definition}\label{def:super-type}
	In a \tOrd $\T=v_1,\ldots,v_n$ and for $1\le i<j\le n$, $\T[i,j]$ is a \emph{super-type}~$s_{i,j}$ of potential~$p\in \N^\Delta$ if $p^\T_{i-1}=p^{\T}_j=p$ and all potentials from position~$i$ till~$j-1$  are different from~$p$.
\end{definition}
Note that, by the definition of super-types, cutting out any super-type from the \tOrd results, by \autoref{lem:cut-out-neutral-block}, again in a \tOrd.
We use this fact later in order to reorder \tOrd{s}.

\begin{definition}\label{def:repetitions}
	A \emph{$k$-repetition} of a super-type~$s$ in a \tOrd~$\phi$ is a subsequence~$\psi$ of~$\phi$ with~$\psi = s^k$, that is, $k$ subsequent occurrences of~$s$. If $k$ is maximal under this condition, then it is called a \emph{maximal $k$-repetition}.
\end{definition}
%
%
Since in the low potential case the values of the occurring potentials in any \TopOr~$\T$ is upper-bounded by~$\Delta^2$, it follows that the number of different occurring potentials is upper-bounded by~$\Delta^{2\Delta}$.
Hence, there are potentials that occur multiple times in~$\T$.

In \autoref{fig:low potential-example} an example for a realizable degree sequence is given where the only existing \TopOr consists basically of one big $k$-repetition of the super-type~$\binom{2}{1},\binom{3}{4}$ of potential~$(2,2,1,1)^T$.
To solve such instances, the algorithm works in two steps.
First, it guesses a so-called \emph{\nRepOr}, that is a \TopOr where \emph{all} maximal $k$-repetitions are replaced by one occurrence of the corresponding super-type.
As one can see in \autoref{fig:low potential-example}, in this example the \nRepOr is very short: $\binom{0}{2},\binom{0}{4},\binom{2}{1},\binom{3}{4}, \binom{2}{0}, \binom{2}{0}, \binom{1}{0}, \binom{1}{0}$, where~$\binom{2}{1},\binom{3}{4}$ is the repeating super-type in a \TopOr.
Indeed, for any degree sequence~$\S$ admitting a low potential realization there exists a ``short'' \nRepOr (see \autoref{lem:number-super-types}).
Then, the \nRepOr can be computed by exhaustive search and, afterwards, the algorithm computes, based on an ILP  (integer linear program) formulation, the missing~$k$-repetitions (see \autoref{lem:SubstringComputingFPT}).
Next, we formalize the idea of \nRepOr{s}.

\begin{definition}
	Let~$\T =v_1, \ldots, v_n$ be a \TopOr for a \DAGR instance~$\S$. The ordering~$\T'$, that results from~$\T$ by replacing each maximal $k$-repetition, $2 \le k \le n$, by a 1-repetition of the corresponding super-type is called  \emph{\nRepOr}.
\end{definition}
Given a \nRepOr~$\T'$ and a \DAGR instance~$\S$, we say a \tOrd~$\T$ \emph{respects}~$\T'$ if~$\T'$ results from replacing all maximal $k$-repetitions by 1-repetitions for each super-type in~$\T$.

In order to prove that the length of \nRepOr{s} can be bounded in a function solely depending on~$\Delta$, we need a ``reordering operation'' for \tOrd{s}, similar to the high potential case.
Cutting out any super-type from the \tOrd results, by the definition of super-types and \autoref{lem:cut-out-neutral-block}, in a \tOrd. In the high potential case we have reinserted the cut out parts right behind a high potential.
Since in the low potential case there exists no high potential we need to show another way to insert the parts that we cut out.
Therefore, the next lemma shows that a \pTopOr with input and output potential~$p$ can be reinserted in a \TopOr at any position~$i$ with potential~$p$.

\begin{lemma}\label{lem:introduce-neutral-block}
	Let $\S$ be a degree sequence with a \TopOr $\phi=v_1,\ldots,v_n$. Furthermore, for a degree sequence~$\S'$ let $\phi'$ be a \pTopOr with input and output potential~$p$. Then, for all indices $1\le i\le n$ where $p^\phi_i= p$, the ordering $\phi'' = \phi[1,i]\phi'\phi[i+1,n]$ is a \TopOr for~$\S\uplus \S'$ with $p^\phi_j = p^{\phi''}_{j+|\phi'|}$ for all~$i < j \le n$.
\end{lemma}
\begin{proof}
	Let $\phi=v_1,\ldots,v_n$ be a \TopOr for a degree sequence~$\S$ and let $\phi'$ be a \pTopOr for a degree sequence~$\S'$ with input and output potential~$p\in \N^\Delta$. By \autoref{def:partial-top-order} there are degree sequences~$P^s$ and $P^t$ such that there is a \TopOr~$\phi'_{s,t}$ for $\S'\uplus P^s\uplus P^t$. Furthermore, in $\phi'_{s,t}$ the first (last) vertices correspond to~$P^s$~($P^t$, respectively).

	We show that for the order $\phi[1,i]\phi'\phi[i+1,n]$ there is a dag~$D$ that corresponds to it, and thus is a realization for $\S\uplus \S'$. We first copy from the dag that corresponds to~$\phi$ all arcs between two vertices in~$\phi[1,i]$ and all arcs between two vertices in~$\phi[i+1,n]$ into~$D$. Correspondingly, from the dag that corresponds to $\phi'_{s,t}$ we copy all arcs between two vertices in~$\phi'$ into~$D$. In the following we shall describe how to connect the remaining three ``components'' $\phi[1,i],\phi'$, and $\phi[i+1,n]$ in~$D$.

	For $1\le l\le \Delta$, let $V^l_{\phi[1,i]}$ be the vertices in~$\phi[1,i]$ that have exactly~$l$ neighbors in~$\phi[i+1,n]$. Symmetrically, let $V^l_{P^s}$ be the vertices in~$\phi'_{s,t}[1,|P^s|]$ that have exactly~$l$ neighbors in~$\phi'_{s,t}[|P^s|+1,|\phi'_{s,t}|]$.
	Clearly, since $p=p^\phi_i$, for all $1\le l\le \Delta$ it follows that $|V^l_{\phi[1,i]}|  =|V^l_{P^s}|$ and thus there is a bijection $h:\left(\bigcup_{l=1}^{\Delta}V^l_{P^s}\right)\rightarrow \left(\bigcup_{l=1}^{\Delta}V^l_{\phi[1,i]}\right)$ such that for $v\in V^l_{P^s}$ it holds that $f(v)\in V^l_{\phi[1,i]}$.
	We now connect the two ``components'' $\phi[1,i]$ and~$\phi'$ by adding for each arc $(u,v)$ with $u\in \phi'_{s,t}[1,|P^s|]$ and $v\in \phi'$ the arc $(h(u),v)$ to~$D$. Since~$h$ is a bijection it is clear that we have not introduced parallel arcs and the indegrees of the vertices~$\phi'$ in~$D$ now matches its element entries in~$\S'$.

	We now consider in~$D$ the vertices in~$\phi[1,i]\phi'$ whose outdegree is less than the outdegree in their corresponding element entry in~$\S\uplus \S'$. We denote this as the \emph{outgoing gap} of such a vertex. The set of vertices with outgoing gap greater than zero
	is a subset of the vertices in $\bigcup_{i=1}^l V^l_{\phi[1,i]}$ and the vertices in $\phi'$ that have in~$\phi'_{s,t}$ neighbors in~$\phi'_{s,t}[|P^s|+|\phi'|+1,|\phi'_{s,t}|]$. We will use these vertices to connect the two remaining ``components'' $\phi[1,i]\phi'$ and~$\phi[i+1,n]$. More specifically, for each $1\le l \le \Delta$ let $V^l_{\phi[1,i]\phi'}$ be the vertices in~$\phi[1,i]\phi'$ with outgoing gap exactly~$l$. Since~$h$ is a bijection and the output potential of~$\phi'$ is~$p$, it follows that there is a bijection $g:\left(\bigcup_{i=1}^l V^l_{\phi[1,i]}\right)\rightarrow \left(\bigcup_{i=1}^l V^l_{\phi[1,i]\phi'}\right)$. Thus, for every arc $(v,u)$ with $v\in \phi[1,i]$ and $u\in \phi[i+1,n]$ we add the arc $(g(v),u)$ to~$D$. Since~$g$ is a bijection it follows that we have not introduced parallel arcs and the indegrees and outdegrees of all vertices correspond to their elements in~$\S \uplus \S'$. Thus, the resulting dag~$D$ with \tOrd~$\phi''$ is a realization for~$\S\uplus \S'$ and $p^\phi_j = p^{\phi''}_{j+|\phi'|}$ for all~$i < j \le n$.
\end{proof}
Combining the cut operation from \autoref{lem:cut-out-neutral-block} and the insert operation from \autoref{lem:introduce-neutral-block} we arrive at the following lemma describing our ``reordering operation''.

\begin{lemma} \label{lem:move-neutral-block}
	Let $\S$ be a degree sequence with a \TopOr $\phi=v_1,\ldots,v_n$. Let~$1 \le i < j < k \le n$ be three positions with~$p_i=p_j=p_k$.
	Then the ordering $\phi[1,i]\phi[j+1,k]\phi[i+1,j]\phi[k+1,n]$ is a \TopOr for~$\S$.
\end{lemma}
With \autoref{lem:move-neutral-block}, it is easy to see that we can reorder any \tOrd~$\T$ such that there is only one consecutive occurrence of a super-type~$s$, that is, beside on maximal $k$-repetition of~$s$ there are no further occurrences of~$s$ in~$\T$.

Next, we show that we can bound the number and the length of super-types in a \nRepOr~$\T'$ for a \DAGR instance~$\S$ by some function only depending on the parameter~$\Delta$.
This allows to determine the \nRepOr by brute force in running time only depending on~$\Delta$.

\begin{lemma}\label{lem:number-super-types}
	Let~$\S$ be a realizable \DAGR instance. If there is a low potential realization for~$\S$, then there exists a \nRepOr~$\T'$ for~$\S$ and a \TopOr~$\T$ for~$\S$ that respects~$\T'$ such that the length of any repeating super-type in~$\T$ is bounded by~$\Delta^{2\Delta}$ and the length of~$\T'$ is bounded by~$\Delta^{2\Delta}(\Delta^{2\Delta^{2\Delta}+2\Delta}+\Delta^{2\Delta})$.
\end{lemma}
\begin{proof}
	Let~$\T =v_1, \ldots, v_n$ be a low potential \TopOr for~$\S$. Let~$\P = p_0, \ldots, p_n$ be the corresponding sequence of potentials with values strictly less than~$\Delta^2$.
	By definition~$p_0=p_n=0^\Delta$. 
	We now construct a \nRepOr~$\T'$ from~$\T$.

	Let~$p$ be a potential. We denote by~${\rm first}_\P(p)$ the position of the first occurrence of~$p$ in~$\P$.
	Let~$p$ be the potential with~${\rm first}_\P(p) > {\rm first}_\P(q)$ for all potentials~$q \ne p$, and $p$ and~$q$ occur in~$\P$.
	That is, $p$~is the potential that occurs at last.
	Let~$i_1, \ldots, i_\ell$ be the occurrences of~$p$ in~$\P$.
	We now show how to construct a \nRepOr with~$i_j - i_{j-1} \leq \Delta^{2 \Delta}$ for all~$1 < j \leq \ell$.

	Assume that there is a~$j \in \{2, \ldots, \ell\}$ such that~$i_j - i_{j-1} > \Delta^{2 \Delta}$.
	Since we assume in this subsection that the value of any occurring potential is lower than~$\Delta^2$, this gives less than~$\Delta^{2 \Delta}$ possible potentials.
	Hence, there are two positions $i_{j-1}<h_1<h_2<i_j$ such that $p_{h_1}=p_{h_2}=:q$.
	Since ${\rm first}_\P(p) > {\rm first}_\P(q)$ there is a position~$k$ with~$k < i_1$ with $p_k=q$.
	By \autoref{lem:move-neutral-block}, $\T[1,k] \cdot \T[h_1+1,h_2] \cdot \T[k+1, h_1] \cdot \T[h_2+1,n]$ is also a \TopOr.
	After exhaustively applying this procedure we have a \TopOr~$\widetilde{\T}$ where the occurrences of~$p$ are~$\widetilde{i}_1, \ldots, \widetilde{i}_\ell$ with~$\widetilde{i}_{j} - \widetilde{i}_{j-1} < \Delta^{2\Delta}$.
	By the same argument, we can assume that~$n-\widetilde{i}_\ell < \Delta^{2\Delta}$ since every potential occurs at most once in~$\widetilde{\T}[\widetilde{i}_l,n]$.
	Thus, there are at most~$O( \Delta^{2\Delta^{2\Delta}})$ many different super-types of potential~$p$.
	Since these super-types are of the same potential~$p$, they can be reordered by using \autoref{lem:move-neutral-block} such that there is at most one subsequent occurrence of each super-type of potential~$p$.
	Thus, in a \nRepOr{} each super-type of potential~$p$ occurs at most once and, hence, the part with the super-types of potential~$p$ has length at most~$O(\Delta^{2\Delta^{2\Delta}} \Delta^{2\Delta})= O(\Delta^{2\Delta^{2\Delta}+2\Delta})$.
	Hence, the length of~$\widetilde{\T}[\widetilde{i}_1, n]$ can be upper-bounded by $O(\Delta^{2\Delta^{2\Delta}+2\Delta} + \Delta^{2\Delta})$.

	So far, the longest repeating super-type is of potential~$p$ and of length at most~$\Delta^{2\Delta}$.
	It remains to bound the length of~$\widetilde{\T}[1, \widetilde{i}_1-1]$.
	Note that the potential~$p$ does not occur in~$\widetilde{\T}[1, \widetilde{i}_1-1]$.
	Now, we can iteratively apply this procedure to deal with the other potentials in~$\widetilde{\T}[1, \widetilde{i}_1-1]$.
	By iteratively applying the described procedure for the at most~$\Delta^{2\Delta}$ different potentials, the length of the resulting \nRepOr is at most~$O(\Delta^{2\Delta}(\Delta^{2\Delta^{2\Delta}+2\Delta}+\Delta^{2\Delta}))$.

	In each iteration the repeating super-types that are deleted through the \nRepOr notion is of the particular potential dealt with in the iteration.
	Hence, the length of the longest such deleted repeating super-type is at most~$\Delta^{2\Delta}$.
	Thus, by reverting all the deletion steps, we get a \TopOr for~$\S$ such that the length of the longest repeating super-type is at most~$\Delta^{2\Delta}$.
\end{proof}
Using \autoref{lem:number-super-types}, the algorithm branches in all possibilities for \nRepOr{s} of length at most~$O(\Delta^{2\Delta}(\Delta^{2\Delta^{2\Delta}+2\Delta}+\Delta^{2\Delta}) )< \Delta^{3\Delta^{2\Delta}}$ for sufficiently large~$\Delta$.
This gives at most~$\Delta^{2\Delta^{3\Delta^{2\Delta}}}$ cases.
\autoref{lem:ordering-defines-potential} shows that by well-connecting the corresponding vertices one can easily check whether the \nRepOr~$v_1, \ldots, v_\ell$ is a \TopOr for the degree sequence~$\left\{\binom{d^-(v_1)}{d^+(v_1)}, \ldots, \binom{d^-(v_\ell)}{d^+(v_\ell)}\right\}$.
For a \nRepOr, the algorithm checks which of the~$(\Delta^2)^{\Delta^{2\Delta}} = \Delta^{2\Delta^{2\Delta}}$ possibly repeating super-types occur in the sequence of the particular case and stores the occurring ones in a set~$\superTypes$.
Then, given a \nRepOr~$\T'$ for a \DAGR instance~$\S$ and the set of~$\superTypes$ super-types that may repeat, the problem of computing a \TopOr that respects~$\T'$ is fixed-parameter tractable with respect to the number of super-types in~$\superTypes$. 
We shall show an ILP formulation for this problem. 
To this end, we formalize the problem and call it \SF.

\begin{center}
\begin{minipage}{0.95\textwidth}
  \defDecprob{\SF}
      {A multiset~$\S = \left\{ \binom{a_1}{b_1}, \ldots, \binom{a_n}{b_n} \right\}$, a \nRepOr~$\T'$ and a set of all super-types $\superTypes = \left\{ s_1, \ldots, s_\ell \right\}$ of~$\T'$ that have length at most~$\Delta^{2\Delta}$.} 
      {Is there a \TopOr~$\T$ for~$\S$ that respects~$\T'$ such that~$\T'$ results from replacing for each~$s \in \superTypes$ all maximal $k$-repetitions of~$s$ in~$\T$ by one occurrence of~$s$?}
\end{minipage}
\end{center}
Next, we show fixed-parameter tractability of \SF with respect to the parameter~$|\superTypes| = k$.
Since the number~$k$ of super-types is bounded by a function only depending on~$\Delta$, this completes our algorithm for the case that an input of \DAGR only admits a low potential realization.

\begin{lemma} \label{lem:SubstringComputingFPT}
	\SF is fixed-parameter tractable with respect to the parameter~$|\superTypes| = k$.
\end{lemma}
\begin{proof}
	We show the fixed-parameter tractability by giving an ILP-formulation of the problem with~$O(k)$ variables.
	\citet{Len83} proved that ILP with~$p$ variables can be solved in $O(p^{4.5p}\cdot L)$ time where~$L$ is the input size.

	To solve the \SF instance, we use the following ILP-formulation:
	\begin{align}
		\forall 1 \leq i \leq k: && \text{\,} f_i & \ge 0 \label{eq:fiGreaterZero} \\
		\forall e \in \S: && \sum_{i = 1}^{k}  f_i \cdot o(e,s_i) & = o(e,\S) - o(e, \phi') \label{eq:elemenSum}
	\end{align}
	Here, the function~$o(e,M)$ denotes the number of occurrences of the element~$e$ in the multiset (or sequence)~$M$.
	The ILP uses the~$k$ integer variables~$f_1, \ldots, f_k$ and consists of~$k + |\S|$ equations that have together a size of~$O(k \cdot |\S|)$.
	Hence, the ILP can be solved in~$O(k^{4.5k}\cdot k\cdot|\S|)$ time.
	
	Now we describe how we use the solution of the ILP-formulation to create a \TopOr~$\T$ as described in the problem definition of \SF.
	For each super-type~$s_i \in \superTypes$ with~$f_i > 0$ insert a~$f_i$-repetition of~$s$ right after an occurrence of~$s$ in~$\T'$.
	By \autoref{lem:introduce-neutral-block}, the ordering that results from adding the $k$-repetitions results in a \TopOr for~$\T' \uplus \biguplus_{i, f_i > 0}\biguplus_{j=1}^{f_i} s_i = \S$.
	
	Next, we show that if the \SF-instance is a yes-instance, there exists a solution for our ILP-formulation.
	If there is a \TopOr~$\T$ for~$\S$ that respects~$\T'$, then there exists a set of super-types~$T_S$ such that replacing all maximal~$k$-repetitions of super-types~$s_i$ in~$T_S$ in~$\T$ by one occurrence of the corresponding super-type results in~$\T'$.
	Clearly,~$T_S \subseteq \superTypes$.
	Set~$f_i = 0$ for each~$s_i \in \superTypes \backslash T_S$ and for all~$s_i\in T_S$ set~$f_i=k-1$ where~$\T$ contains a maximal $k$-repetition of~$s$.
	Thus, Inequality~\eqref{eq:fiGreaterZero} is fulfilled for all~$1\le i \le k$.
	Since~$\T$ is a \TopOr for~$\S$, Inequality \eqref{eq:elemenSum} is fulfilled. 
\end{proof}
Combining \autoref{lem:number-super-types} and \autoref{lem:SubstringComputingFPT} shows fixed-parameter tractability for the low potential case. 
The running time in this case is~$O(\Delta^{2\Delta^{3\Delta^{2\Delta}}})$ for checking all \nRepOr{s}, $O(\Delta^{2\Delta^{2\Delta}}\cdot\Delta^2\cdot n)$ for constructing the \SF instance, and~$O(\Delta^{2\Delta^{2\Delta} 6\Delta^{2\Delta^{2\Delta}}} \cdot \Delta^{2\Delta^{2\Delta}}\cdot n)$ for solving the ILP.
Altogether we have the following theorem.

\begin{theorem} \label{thm:lowPotential}
	If a degree sequence admits a low potential \TopOr, then it can be found in~$\Delta^{\Delta^{\Delta^{O(\Delta)}}}\cdot n$ time.
\end{theorem}
\theoremautorefname{s} \ref{thm:highPotential} and \ref{thm:lowPotential} together lead to the main theorem of this section.

\begin{theorem}
	\DAGR is fixed-parameter tractable with respect to the parameter maximum degree~$\Delta$.
\end{theorem}
Note that this is a mere classification result: The running time is~$\Delta^{\Delta^{\Delta^{O(\Delta)}}}\cdot n$. It is dominated by the low potential case.

\section{Conclusion and Open Questions}

Answering an open question by \citet{BM11} we proved the NP-completeness of \DAGR. 
Following the spirit of deconstructing intractability we figured out the necessity of large degrees in the NP-hardness proof by showing fixed-parameter tractability for \DAGR with respect to the maximum degree~$\Delta$.
The natural questions whether \DAGR is solvable in single-exponential time and whether it admits a polynomial-size problem kernel with respect to the parameter~$\Delta$ arises.
In our NP-hardness reduction other parameters occur with unbounded values, for instance, the number of types. 
Investigating this parameter is an interesting task for future work.

\paragraph{Acknowledgements.} 
Our thanks for for fruitful discussions about \DAGR goes to Annabell Berger and Matthias Müller-Hannemann. Moreover, we are grateful to Rolf Niedermeier for inspiring discussions and helpful comments improving the presentation.

{\footnotesize
\bibliographystyle{abbrvnat}
\bibliography{dag}
}

\end{document}